\def\ShowAuthNotes{0}
\def\llncs{0}
\def\anon{0}

\ifnum\llncs=0
\documentclass[11pt]{article}
\usepackage{fullpage}
\else
\documentclass[orivec,runningheads,envcountsect,envcountsame]{llncs}
\usepackage{amssymb,amsmath,cite,url,enumerate,algorithm,algorithmic,latexsym}
\usepackage[table]{xcolor}
\fi

\usepackage{latexsym}
\usepackage{amssymb}
\usepackage{amsmath}
\usepackage{color}
\usepackage{tabularx}
\usepackage{url}
\definecolor{DarkBlue}{RGB}{0,0,150}
\usepackage[colorlinks,linkcolor=DarkBlue,citecolor=DarkBlue]{hyperref}

\usepackage{array}
\usepackage{epsfig}

\usepackage{bm}
\usepackage{bbm}

\ifnum\ShowAuthNotes=1
\newcommand{\authnote}[3]{\textcolor{#3}{[{\footnotesize {\bf #1:} { {#2}}}]}}
\else
\newcommand{\authnote}[3]{}
\fi

\newcommand{\omri}[1]{\authnote{Omri}{#1}{blue}}
\newcommand{\znote}[1]{\authnote{Z}{#1}{red}}

\newcommand{\PRF}{\mathsf{PRF}}

\newcommand{\Keys}{\mathcal{K}}

\ifnum\llncs=0

\usepackage{amsthm}

\newtheorem{theorem}{Theorem}[section]
\newtheorem{definition}[theorem]{Definition}

\newtheorem{lemma}[theorem]{Lemma}

\newtheorem{corollary}[theorem]{Corollary}

\fi

\newtheorem{cclaim}[theorem]{Claim} 
\newtheorem{fact}[theorem]{Fact}

\def\bbC{{\mathbb C}}
\def\bbE{{\mathbb E}}

\def\bbN{{\mathbb N}}

\def\bbR{{\mathbb R}}

\def\bbZ{{\mathbb Z}}

\newcommand{\ceil}[1]{\left\lceil #1 \right\rceil}

\def\binset{\{0,1\}}
\def\pmset{\{\pm 1\}}

\newcommand{\abs}[1]{\left\vert {#1} \right\vert}
\newcommand{\norm}[1]{\left\| {#1} \right\|}
\newcommand{\td}{\mathrm{TD}}
\newcommand{\sd}{\mathrm{SD}}

\newcommand{\fail}{\mathtt{fail}}
\newcommand{\success}{\mathtt{success}}

\newcommand{\SingleDistribution}[1]{\mathcal{N}^{\bbC}_{\text{R}\left( {#1} \right)}(0, 1)}
\newcommand{\normalDist}{\mathcal{N}(0, 1)}
\newcommand{\normalCompDist}{\mathcal{N}^{\bbC}(0, 1)}
\newcommand{\distribution}[2]{\big(\mathcal{N}^{\bbC}_{\text{R}\left( {#2} \right)}(0, 1) \big)^{{#1}}}
\newcommand{\gauss}[1]{G^{\bbC}_{\text{R}\left( {#1} \right)}}

\newcommand{\qrs}{\mathsf{QRS}}
\newcommand{\bvs}{\mathsf{BVS}}

\newcommand{\Gen}{\mathsf{Gen}}
\newcommand{\hatv}{\hat{v}}
\newcommand{\hatu}{\hat{u}}
\newcommand{\hatphi}{\hat{\varphi}}

\def\poly{{\rm poly}}

\def\negl{{\rm negl}}

\newcommand{\A}{\mathsf{A}}

\newcommand{\secp}{\lambda}
\newcommand{\Img}{\text{Img}}
\newcommand{\PreImg}{\text{PreImg}}

\newcommand{\ket}[1]{|{#1}\rangle}
\newcommand{\bra}[1]{\langle{#1}|}
\newcommand{\braket}[2]{\langle{#1}|{#2}\rangle}
\newcommand{\ketbra}[2]{\ket{{#1}}\bra{{#2}}}

\newcommand{\boldpar}[1]{\vspace{3pt}\par\noindent\textbf{#1}}

\ifnum\llncs=1
\renewcommand{\paragraph}{\boldpar}
\fi

\title{Scalable Pseudorandom Quantum States}

\ifnum\llncs=1

\ifnum\anon=1
\author{}
\institute{}
\else

\author{Zvika Brakerski\inst{1} \and Omri Shmueli\inst{2}}
\institute{Weizmann Institute of Science\thanks{Email: \texttt{zvika.brakerski@weizmann.ac.il}. Supported by the Binational Science Foundation (Grant No.\ 2016726), and by the European Union Horizon 2020 Research and Innovation Program via ERC Project REACT (Grant 756482) and via Project PROMETHEUS (Grant 780701).} \and Tel-Aviv University\thanks{Email: \texttt{omrishmueli@mail.tau.ac.il}. Supported by the European Union Horizon 2020 Research and Innovation Program via ERC Project REACT (Grant 756482), by the Israel Science Foundation Grant No. 18/484, and by Len Blavatnik and the Blavatnik Family Foundation.}}

\fi

\else
\ifnum\anon=1
\author{}
\else

\author{Zvika Brakerski\thanks{Weizmann Institute of Science, \texttt{zvika.brakerski@weizmann.ac.il}. Supported by the Binational Science Foundation (Grant No.\ 2016726), and by the European Union Horizon 2020 Research and Innovation Program via ERC Project REACT (Grant 756482) and via Project PROMETHEUS (Grant 780701).} \and Omri Shmueli\thanks{Tel-Aviv University, \texttt{omrishmueli@mail.tau.ac.il}. Supported by the European Union Horizon 2020 Research and Innovation Program via ERC Project REACT (Grant 756482), by the Israel Science Foundation Grant No. 18/484, and by Len Blavatnik and the Blavatnik Family Foundation.}}
\date{}

\fi

\fi

\begin{document}
	\maketitle

\begin{abstract}
	Efficiently sampling a quantum state that is hard to distinguish from a truly random quantum state is an elementary task in quantum information theory that has both computational and physical uses. This is often referred to as pseudorandom (quantum) state generator, or PRS generator for short.
	
	In existing constructions of PRS generators, security scales with the number of qubits in the states, i.e.\ the (statistical) security parameter for an $n$-qubit PRS is roughly $n$. Perhaps counter-intuitively, $n$-qubit PRS are not known to imply $k$-qubit PRS even for $k<n$. Therefore the question of \emph{scalability} for PRS was thus far open: is it possible to construct $n$-qubit PRS generators with security parameter $\secp$ for all $n, \secp$. Indeed, we believe that PRS with tiny (even constant) $n$ and large $\secp$ can be quite useful.
	
	We resolve the problem in this work, showing that any quantum-secure one-way function implies scalable PRS. We follow the paradigm of first showing a \emph{statistically} secure construction when given oracle access to a random function, and then replacing the random function with a quantum-secure (classical) pseudorandom function to achieve computational security. However, our methods deviate significantly from prior works since scalable pseudorandom states require randomizing the amplitudes of the quantum state, and not just the phase as in all prior works. We show how to achieve this using Gaussian sampling.
\end{abstract}

	\thispagestyle{plain} 
	\pagestyle{plain} 

\section{Introduction}

Quantum mechanics asserts that the state of a physical system is characterized by a vector in complex Hilbert space, whose dimension corresponds to the number of degrees of freedom of the system. Specifically, a system with $2^n$ possible degrees of freedom (such as an $n$-qubit system, the quantum analogue to an $n$ bit system) is represented as a unit vector over $\bbC^{2^n}$. The ability to sample a random state of a system is a fundamental task when attempting to provide a computational description of the physical world.

Since the description length of a quantum state is infinite (and very long even when taken to a finite precision), relaxed notions for random state sampling are considered in the literature. Most commonly (and in this work) we consider restricting the \emph{number of copies} of the sampled state that are given to the adversary.\footnote{Recall that in the quantum setting, due to the no-cloning property, providing additional copies of the same state allows to recover more information about it. In utmost generality, any additional copy provides additional information, and a complete recovery of a quantum state requires infinitely many copies.} The notion of quantum $t$-designs \cite{ambainis2007quantum} considers computationally unbounded adversaries that are given $t$ copies of the sampled state, and the requirement is that this input is (statistically) indistinguishable from $t$ copies of a true random state. The resources of generating $t$-designs scale at least linearly with $t$, and therefore if efficient generation is sought, $t$ designs can only be constructed for polynomial $t$.\footnote{As usual, we use the notion of security parameter $\secp$ that indicates the power of honest parties and of adversaries. We assume that honest parties run in time $\poly(\secp)$ for a \emph{fixed} polynomial, whereas the advantage of the adversary needs to scale super-polynomially, and preferably exponentially, with $\secp$.} Recently, a computational variant known as Pseudorandom Quantum State (PRS) was proposed by Ji, Liu and Song~\cite{JLS18}. In a PRS, the adversary is allowed to request an a-priori unbounded polynomial number of samples $t$, but the guarantee of indistinguishability only holds against \emph{computationally bounded} adversaries. PRS have applications in quantum-cryptography (e.g. quantum money \cite{JLS18}) and computational physics (e.g. simulation of thermalized quantum states \cite{popescu2006entanglement}).

It was shown in \cite{JLS18,BS19} that PRS can be constructed from any quantum-secure one-way function. The design paradigm in both works is as follows. First, assume you are given (quantum) oracle access to a (classical) random function, and show how to efficiently construct a PRS which is secure even against computationally unbounded adversaries, a notion that \cite{BS19} calls Asymptotically Random State (ARS). Then, replace the random function with a post-quantum pseudorandom function (PRF) to obtain computational security. Since only a fixed number of calls to the PRF is required in order to generate each PRS copy, this paradigm also leads to new constructions of $t$-designs, as observed in \cite{BS19}.

The previous works \cite{JLS18,BS19} showed how to construct an $n$-qubit PRS, which is secure against any $\poly(n)$ time adversary. To be more precise, they constructed ARS whose distinguishing advantage is bounded by $4 t^2 \cdot 2^{-n}$, and converted it into a PRS using a PRF as described above. We can therefore say that the \emph{statistical security parameter} of the scheme is (essentially) $n$, and there is an additional computational security parameter that comes from the hardness of the PRF. Indeed, a security parameter of $n$ seems quite sufficient since the complexity of the construction is $\poly(n)$ so it is possible to choose $n$ as large as needed in order to provide sufficient security. Alas it is not possible to convert an $n$-qubit state generator into one that produces a random state over a smaller number of qubits, say $k<n$. This may be quite surprising as one would imagine that we can simply generate an $n$-qubit state, and just take its $k$-qubit prefix. However, recall that the $n$-qubits are in superposition, and taking a prefix is equivalent to measurement of the remaining $(n-k)$ qubits. For each of the $t$ copies, this measurement has a different outcome and therefore each of the $t$ copies will produce a different $k$-qubit states, as opposed to $t$ copies of the same state as we wanted.

This peculiar state of affairs means that prior to this work it was not known, for example, how to construct ARS/PRS of $n$ qubits, but with adversarial advantage bounded by $2^{-2n}$. This issue is also meaningful when considering the concrete (non-asymptotic) security guarantees of PRS, where we wish to obtain for example $128$ bits of security against an adversary that obtains at most $2^{20}$ copies of a PRS over $70$ qubits.

\paragraph{This Work: Scalable ARS/PRS.} In this light, it is desirable to introduce ARS/PRS constructions where the security parameter is in fact a parameter which is tunable independently of the length of the generated state. We call this notion \emph{scalable} ARS/PRS. We notice that the approaches of \cite{JLS18,BS19} are inherently not scalable since they can only generate states in which all computational-basis elements have the same amplitude, and the randomness only effects the phase. Such vectors are inherently distinguishable from uniform unless the dimension is very large (hence their dependence between length and security). In this work, we present new techniques for constructing ARS/PRS and in particular present a scalable construction under the same cryptographic assumptions as previous works.

\subsection{Our Results} 
Our main technical result, as in all previous works, is concerned with constructing an ARS generator which is efficient given oracle access to a random function.\footnote{Note that this is not the quantum random oracle model since the random oracle is ``private'' and the adversary does not get access to it.}

\begin{lemma} [Main Technical Lemma]\label{lem:intromain}
	There exists a scalable ARS generator. 
	
	Furthermore, for every length $n$ of a quantum state and security parameter $\secp$, running the generator $t$ times (for any $t$) produces an output distribution that is $O\left( \frac{t}{e^{\secp}} \right)$-indistinguishable from $t$ copies of a random quantum state of $n$ qubits.
\end{lemma}

We note that in previous works that construct ARS generators \cite{JLS18,BS19} the dependence on $t$ in the bound on the trace distance is quadratic, that is, previous ARS generators are known to achieve a bound of $\frac{t^2}{2^n}$ on the trace distance between $t$-copies of the ARS and a random quantum ($n$-qubit) state, whereas in this work the trace distance bound only scales up linearly with $t$.

As immediate corollaries and similarly to \cite{JLS18,BS19}, we derive the existence of a scalable PRS generator (assuming post-quantum one-way functions) and scalable $t$-design generators (unconditionally).
Unlike scalable PRS generators, scalable state $t$-design generators were known to exist before this work, however their depth was known to scale up linearly with $t$ (and polynomially in $n$), and in our construction the depth scales logarithmically with $t$ (and polynomially in $n, \secp$).

\begin{corollary}\label{cor:introprs}
	If post-quantum one-way functions exist, then scalable PRS generators exist.
\end{corollary}

\begin{corollary}\label{cor:introdesigns}
	For any polynomial $t(\cdot) : \bbN \rightarrow \bbN$, scalable state $t(\secp)$-design generators exist where the circuit depth is $\poly(n, \secp, \log t)$.
\end{corollary}

Our ARS construction requires a random oracle with $n$ bits of input (where $n$ is the length of the generated state) and $\poly(\secp)$ bits of output, it therefore follows that if $n=O(\log \secp)$, then it is possible to instantiate the construction with a completely random string of length $2^n \cdot \poly(\secp) = \poly(\secp)$, and obtain statistically secure PRS. We view this consequence as not very surprising in hindsight.

Recently Alagic, Majenz and Russell \cite{AMR19} proposed the notion of random state simulators. Simulators are stateful, and their local state grows with the number of copies $t$, however, there is no a-priori bound on the number of copies that the simulator can produce, and the guarantee is information-theoretic rather than computational. One can observe that a scalable ARS generator also implies efficient state simulators, by using the random-oracle simulation technique of Zhandry~\cite{Z19}. The state simulators of \cite{AMR19} follow a different approach, which is not known to imply ARS, and achieve simulators with perfect security (and thus straightforwardly scalable), but our ARS provides a different avenue for scalable random quantum state simulators as well.

\subsection{Paper Organization}

We provide a detailed technical overview of our results in Section \ref{sec:techoverview}. Preliminaries appear in Section~\ref{sec:prelim}, and in particular we formally state the derivation of the corollaries from the main theorem (which were implicit in previous work) in Section \ref{sec:consequences}.
Our technical results are presented in the following two sections. In Section \ref{sec:qit_tools} we present quantum information-theoretic tools which are required for our construction but may also find other uses.
Then Section \ref{sec:construction} contains our actual construction.

\newcommand{\halpha}{\hat{\alpha}}
\newcommand{\hv}{\hat{v}}
\newcommand{\ketv}{\ket{v}}
\newcommand{\khv}{\ket{\hv}}

\newcommand{\hw}{\hat{w}}
\newcommand{\ketw}{\ket{w}}
\newcommand{\khw}{\ket{\hw}}

\section{Technical Overview}
\label{sec:techoverview}

We now provide a technical outline of how we achieve our main result in Lemma~\ref{lem:intromain}. Deriving the corollaries is straightforward using known techniques.\footnote{We note that this standard transition from ARS with oracle to PRS and to $t$-designs was not formally stated in its generic form in previous works. In this work we also provide the generic derivations in Section \ref{sec:consequences}.}

As Lemma~\ref{lem:intromain} states, we design an algorithm that has oracle access to a random function $f$, takes as input a bit length $n$ and a security parameter $\secp$, runs in time $\poly(n,\secp)$, and produces a quantum state over $n$-qubits $\ket{\psi_{f,n,\secp}}$ (note that even though our algorithm is randomized, it can either output the state $\ket{\psi_{f,n,\secp}}$ or $\bot$ and will never output the ``wrong'' state). It furthermore holds that the distribution that samples a random function $f$ and outputs $\ket{\psi_{f,n,\secp}}^{\otimes t}$ (i.e.\ $t$ copies of the state $\ket{\psi_{f,n,\secp}}$), is within trace distance at most $\poly(t)/2^\secp$ from the distribution that produces $t$ copies of a truly randomly sampled $n$-qubit state.

We recall the standard Dirac notation for vectors in Hilbert space. An $n$-qubit state is generically denoted by a unit vector in $\bbC^{2^n}$ of the form $\ket{\alpha} = \sum_{x \in \binset^n} \alpha_x \ket{x}$. Throughout this overview we wish to refer to normalized as well as non-normalized vectors. We will use the convention that a vector $\ket{\alpha}$ is not necessarily normalized unless explicitly noted that it represents a quantum state (or a unit vector), and will denote its normalization
\[
\ket{\halpha} = \sum_{x \in \binset^n} \halpha_x \ket{x} := \frac{1}{\sqrt{\braket{\alpha}{\alpha}}} \sum_{x \in \binset^n} \alpha_x \ket{x}~,
\]
where $\braket{\alpha}{\alpha} = \sum_x \abs{\alpha_x}^2$.

As explained above, prior works generated quantum states where in the standard basis all coefficients had the same amplitude, i.e.\ their ARS could be represented by $\ket{\alpha}$ s.t.\ $\abs{\alpha_x}=1$ for all $x$. We abandon this approach, which as we explained cannot lead to a scalable ARS construction. Instead, we will show how to interpret a random function $f$ as an implicit representation of a random unit vector in $\bbC^{2^n}$. Moreover, we want this interpretation to be \emph{locally computable} in the sense that the value $\alpha_x$ only depends on $f(x)$.
Our approach, therefore, is more direct and also more involved than the approach taken in previous works, since we will try to sample from a space that most closely resembles the uniform distribution over quantum states.


\subsection{Our Approach: Implicit Random Gaussian Vector}

Assume that we had an efficiently computable classical function $g(\cdot)$ s.t.\ if we set $v_x = g(f(x))$ and consider the vector $\ketv = \sum_{x} v_x \ket{x}$, then the distribution on $\ketv$ (induced by sampling the function $f$ randomly) is \emph{spherically symmetric}, i.e.\ invariant to unitary transformations (``rotations'' in $\bbC^{2^n}$). In this case, the normalized vector $\khv$ is a uniform unit vector. In other words, we will show how to use the random function $f$ as an implicit representation of a vector $\ket{v}$ such that for all $x$, $v_x$ can be efficiently \emph{locally} computed given $x$ (and oracle access to $f$).

Our solution, therefore, needs to address two challenges. The first is to properly define a locally efficiently computable function $g$ with the desirable properties. The second is to efficiently generate the quantum state $\ket{\hat{v}}$ given oracle access to the values $v_x$. Let us describe how we handle each one of these challenges at a high level, and then expand on the parts that contain the bulk of technical novelty.


%


\paragraph{First Technique: Multivariate Gaussian Sampling.} For the first challenge, we use the multivariate Gaussian distribution, whose spherical symmetry has proven useful for many applications in the literature.
Our function $g$ will simply be a Gaussian sampler (or more accurately, a two-dimensional Gaussian sampler, for the real and imaginary parts of $v_x$).  That is, we use the entries of the random function $f$ as random tape for a Gaussian sampling procedure $g$. Since the Gaussian distribution is spherically symmetric, such a $g$ has the properties that we need.

This approach indeed seems quite suitable but achieving (perfect) spherical symmetry is at odds with achieving computational efficiency, simply because the Gaussian distribution is continuous and has infinite support. Indeed, we will need to show a truncated discretized Gaussian distribution which on one hand can be sampled efficiently, and on the other hand provides approximate spherical symmetry. Note that the notion of approximation we are interested in here is with respect to the trace distance between the quantum state $\khv^{\otimes t}$ and a $t$-repetition of a random unit vector. This requires us to develop tools in order to relate this notion to standard notions such as Euclidean distance. These tools are not particularly complicated but we view them as fundamental and of potential to be used elsewhere.\footnote{We will not be surprised if they were already discovered and used in the literature, but we were unable to find a relevant reference.} We elaborate more on this in Section~\ref{sec:overview:apxgaussian} below, and the full details appear in Section \ref{sec:qit_tools}.


\paragraph{Second Technique: Rejection Sampling.}
The second challenge is addressed using a quantum analog of the \emph{rejection sampling} technique. Recall that in standard probability theory, if it is possible to sample from a distribution $p$ where $\Pr[x]=p_x$, then we can consider the experiment of first sampling from $p$, and then either outputting the sample received $x$ with probability $q_x$, or aborting and restarting the process with probability $1-q_x$. This process constitutes a sampler for the distribution $\frac{p_x q_x}{\sum_{x} p_x q_x}$. The probability of not aborting is $\sum_{x} p_x q_x$, and therefore the expected running time of the new sampler is $\frac{1}{\sum_{x} p_x q_x}$. In the quantum setting, a similar technique can be used for superpositions (Indeed, extensions of these technique were used e.g.\ in \cite{ozols2013quantum}).



In this work we use quantum rejection sampling to generate quantum states from scratch.
To create our state $\ketv$ we will start with the uniform superposition $\ket{u} = \sum_x \ket{x}$, and via a rejection process we can obtain (not necessarily with good pobability), any desired superposition $\ket{v}$.
The probability of success in the quantum case is $\frac{1}{d^2} \cdot \frac{\braket{v}{v}}{\braket{u}{u}}$, where $d$ is an \emph{a-priori} bound on $\max_x \abs{v_x}$ that needs to be given as a parameter to the rejection sampling procedure. (The algorithm and success probability are analogous to the classical version described above, when replacing $q_x$ with $\frac{v_x}{d}$ and considering $\ell_2$ norm instead of $\ell_1$.)

On the face of it, the rejection sampling procedure can work to create any state $\ket{v}$ when a bound $d$ is known. However, the probability of success can still be very small (e.g.\ negligible), so if we wish to use repetition to obtain $\ket{v}$, the expected running time will become very large (e.g.\ super-polynomial).
Fortunately, our vectors $\ket{v}$ are (approximately) Gaussian, which means that they have strong concentration properties that guarantee that with high probability two properties are satisfied.
The first is that all entries $v_x$ have roughly the same magnitude, up to a factor of $\poly(n,\secp)$.\footnote{Note that, e.g.\ tail bounds on the norm of a Gaussians asserts that the probability that its amplitude is beyond $k$ times standard deviation is at most $e^{-c \cdot k^2}$ for some constant $c$. This means that if we want to find a tail bound that applies to all $2^n$ components of the vector $\ket{v}$ at the same time via union bound, it suffices to use $k \approx \sqrt{n+\secp}$.} This allows us to choose the value $d$ in such a way that the rejection sampling algorithm will operate correctly.
The second property is that $\braket{v}{v} \approx 2^n$ (formally, $\braket{v}{v}$ is a constant factor away from $2^n$), this makes the probability of success noticeable (i.e.\ $1/\poly(n,\secp)$).
We informally call a vector that maintains the combination of these two properties "balanced".
By running in time $\poly(n, \secp)$ and repeating the process as needed we can amplify the success probability to $1-2^{-\secp}$.
We generalize these properties and provide a state generator for any oracle $v_x$ which satisfied the balance property, see Section~\ref{sec:qit_tools}.

Lastly, we note that while the first property above (bound on $d$) can be made to hold for any $n$, the second one (lower bound on $\braket{v}{v}$) might not hold with high enough probability. Special care needs to be taken in the case where $n$ is very small, since in that case concentration properties are insufficient to imply that $\braket{v}{v}$ does not fall far below its expected value with small yet significant probability (we wish to succeed with all but $2^{-\secp}$ probability, so anything higher than that is already significant). In such a case, the success probability of the rejection sampler might become negligibly small, which will lead to failure in generating a state.\footnote{We stress again that if the success probability becomes negligible with only negligible probability, e.g.\ $2^{-\sqrt{\secp}}$, this is still a problem since the state generator will simply fail with this probability and therefore we cannot hope to be $2^{-\secp}$ close to uniform.} Luckily, since the dimension of the vector $\ketv$ is $2^n$, good concentration kicks in already at $n \gtrsim \log(\secp)$, so we only need to worry about this issue when $n < \log(\secp)$. For such small $n$, the sampling algorithm can store the vector $\ket{v}$ in its entirety, and check whether the norm $\braket{v}{v}$ is sufficiently close to its expectation (which happens with constant probability). If the norm is not in the required range, we sample a new Gaussian.\footnote{Recall that we think of the values of the function $f(x)$ as the random tape of a Gaussian sampler $g$. We can consider a function $f$ with output length which is $\secp$ times the number of random bits used by the sampler $g$, so that we have sufficient randomness to re-run $g$ as needed.} Repeating this roughly $\secp$ times guarantees that we generate a ``balanced'' vector from a spherically symmetric distribution with all but $2^{-\secp}$ probability.

\subsection{Approximate Gaussians Under Tensored Trace Distance}
\label{sec:overview:apxgaussian}

We wish to do approximate sampling from the continuous Gaussian distribution using an efficiently locally sampleable distribution. If we wish to be fully precise, we need to consider Gaussian distributions over the complex regime. However, for the purpose of sampling, one can think of each complex coordinate just as two real-valued coordinates. For the purpose of this overview we will simplify things even further and assume that we wish to sample from a real-valued Gaussian, i.e.\ a vector in $\bbR^{2^n}$ instead of $\bbC^{2^n}$. Everything we discuss here be extended to the complex regime in a natural manner. From this point and on, our goal is to find an efficient sampler $g$ s.t.\ when sampling $v_x$ i.i.d from the distribution generated by $g$, and sampling $w_x$ from a continuous Gaussian, it holds that the trace distance (quantum optimal distinguishing probability) between the quantum states $\khv^{\otimes t}$ and $\khw^{\otimes t}$, is at most $\poly(t) \cdot 2^{-\secp}$ for all $t$. For any vectors $\ketv, \ketw$, we refer to the trace distance between $\khv^{\otimes t}$ and $\khw^{\otimes t}$ as the ``$t$-tensored trace distance'' between $\ketv$ and $\ketw$.

An efficiently sampleable distribution is necessarily discrete and supported over a finite segment, whereas the Gaussian distribution is continuous and supported over $(-\infty, \infty)$. Indeed, even in the classical setting Gaussian samplers need to handle this discrepancy. Usually, when one says that it is efficient to sample from the Gaussian distribution, they mean that it is possible to sample to within any polynomial precision and from a Gaussian truncated far enough away from the standard deviation that the probability mass that is chopped off is negligible.\footnote{An alternative to chopping the ends of the distribution is to construct a sampler that runs only in expected polynomial time and might run for a very long superpolynomial time with small probability. This approach is less suitable for our purposes.} We adopt a similar approach here. Formally, sampling to within a fixed precision is equivalent to sampling from a \emph{rounded} Gaussian distribution, i.e.\ the distribution obtained by sampling from a continuous Gaussian and then rounding the result to the nearest multiple of $\epsilon$, where $\epsilon$ indicates the required precision. Truncation means that we sample from the distribution obtained by sampling a Gaussian, and if the absolute value of the sampled value $x$ is at most some bound $B$, then return $x$, otherwise return $0$.
Setting $B$ to be sufficiently larger than the standard deviation, say by roughly a factor of $k$, would imply that the resulting distribution only distorts the Gaussian by $e^{-k^2}$ in total variation distance.
We set our sampler $g$ therefore to be a sampler from the $B$-truncated $\varepsilon$-rounded Gaussian distribution.
It is possible to sample from a distribution that's within $\varepsilon$ statistical distance from this distribution in time $\poly(\log(1/\epsilon), \log(B))$ by standard Gaussian sampling techniques, and therefore we can set $1/\epsilon$ to be a sufficiently large exponential function in $\secp, n$ and maintain the efficient sampling property.

The challenge, as already mentioned above, is to translate this intuitive notion of ``approximate Gaussian'' to one that is provable under tensored trace distance. In fact, we present a general analysis of the effects of truncation and rounding on tensored trace distance. We do this using a two-phase proof.

\paragraph{Part I: Tensored Trace Distance Respects Statistical Distance.}
We show that truncating a continuous Gaussian introduces negligible trace distance for \emph{any} number of copies $t$. This follows quite straightforwardly from the classical total variation distance bound between the distributions. In fact, we show a more general claim (Lemma~\ref{lemma:sd_td}): Let $\ketv$ and $\ketw$ be distributions over $n$-qubit states, such that their classical distributions as $2^n$-dimensional vectors are within classical statistical distance (total variation distance) $\delta$. Then their $t$-tensored trace distance is at most $\delta$ for all $t$. The intuition here (which can also be translated to a formal proof), is that even given an infinite number of repetitions, a quantum state does not contain more information than its $2^n$-dimensional coefficient vector. Therefore, a (computationally unbounded) adversary that attempts to distinguish $\khv^{\otimes t}$ and $\khw^{\otimes t}$ as quantum states cannot do better than a classical (computationally unbounded) adversary which receives $\ketv, \ketw$ as explicit vectors.

\paragraph{Part II: Tensored Trace Distance Respects Rounding.} We say that a distribution $\ketv$ is a rounding of a distribution $\ketw$ if $\ketv$ can be described as first sampling an element from $\ketw$ and then applying some mapping $\varphi$ s.t.\ for all $w$, $\norm{\varphi(w)-w}$ is bounded (say be some value $\delta$).\footnote{Note that we call this ``rounding'' but in general this can be applied in other situations.} We wish to show that if $\ketv$ is a rounding of $\ketw$ then these vectors are close under tensored trace distance. 

Let us start by considering the case $t=1$, i.e.\ the distinguisher needs to distinguish between the quantum states $\khv$ and $\khw$. It is well established that if $\khv$ and $\khw$ are close in Euclidean distance, then they are also close in trace distance. However, this does not complete the proof since we only have a bound on the Euclidean distance between the unnormalized vectors $\ketv$ and $\ketw$. Indeed, the notion we care about is the Euclidean distance when projected onto the unit sphere, or in other words the \emph{angular} distance induced by $\varphi$. In our case, our distribution $\ketw$ (the Gaussian) is such that the norm is quite regular with high probability, and this is preserved also for the rounded version (some straightforward yet fairly elaborate calculation is required in order to establish the exact parameters).\footnote{This introduces an additional layer of complication into our proof, as we will need to apply the rounding tool to a restriction of the Gaussian distribution for which the norm is well behaved. Since the ``regular norm'' variant is close in statistical distance to the standard Gaussian, this can be handled by our first technique above.}

Once we formalize the right notion of approximation (i.e.\ angular distance), it is possible to state a general lemma (Lemma~\ref{lemma:angular}) that shows that if $\varphi$ is s.t.\ the angular distance between its input and output (over the support of $\ketv$) is bounded, then the $t$-tensored trace distance degrades moderately with $t$. Therefore, if we start with a short enough angular distance, our trace distance will indeed be bounded by $\poly(t)/2^{\secp}$.

\section{Preliminaries}
\label{sec:prelim}

\subsection{Standard Notions and Notations}
During this paper we use standard notations from the literature.
For $n \in \bbN$,
\begin{itemize}
	\item
	We denote $[n] := \{1, \cdots, n \}$.
	
	\item 
	We denote by $[n]_2$ the $\ceil{\log_2(n)}$-bit binary representation of $n$.
	
	\item
	We denote by $\omega_{n}$ the complex root of unity of order $n$: $\omega_{n} := e^{\frac{2\pi i}{n}}$.
	
	\item
	We denote by ${\cal S}(n)$ the set of $n$-qubit pure quantum states, by ${\cal D}(n)$ the set of $n$-qubit mixed quantum states and by ${\cal U}(n)$ the set of $n$-qubit quantum unitary circuits.
	
	
	\item
	We sometimes denote $2^n$ with $N$, when we do that, we explicitly note it.
	
\end{itemize}

\paragraph{Vectors and Quantum States.}
We use standard Dirac notation throughout this paper, vectors are not assumed to be normalized unless explicitly mentioned.
 Specifically, for a column vector $u \in \bbC^{m}$, we denote $\ket{u} := u$, $\bra{u} := u^{\dagger}$, where $u^{\dagger}$ is the conjugate transposed of $u$. 
We usually let $\hat{u}$ denote the normalized version of the vector $u$, namely: $\hat{u} := \frac{1}{\norm{u}} \cdot u$ (where $u$ is a nonzero complex vector). Vectors that represent quantum states have unit norm and therefore are normalized by default.

We make a distinction between a \emph{vector} in a Hilbert space, and the \emph{quantum state} corresponding to this vector. The two objects are related as a complete characterization of a (pure) quantum state over $n$-qubits is characterized by a vector in a $2^n$-dimensional Hilbert space (up to normalization and global phase). However, the vector is not necessarily (and almost always is not) recoverable given the $n$-qubit state, and quantum states that correspond to different vectors can be indistinguishable (even perfectly).\footnote{Information theoretically, in the general case, one requires an infinite number of copies of a quantum state in order to precisely recover the vector in the Hilbert space that characterizes this state.} In terms of vector notation, the symbol $\ket{u}$ can refer either to the vector in the Hilbert space of to the quantum state that corresponds to this vector, we will explicitly mention which of the two we refer to when using this notation.

%
%

\paragraph{Distributions Over Quantum States as Density Matrices.}
Density matrices are a mathematical tool to describe mixed quantum states, that is, distributions over quantum states.
Formally, let $\mu$ a (possibly continuous) probability distribution over $n$-qubit quantum states, $\mu : {\cal S}(2^n) \rightarrow [0, 1]$, $\int_{\ket{\psi} \in {\cal S}(2^n)} 1 \text{d}\mu(\ket{\psi}) = 1$, then the density matrix induced by $\mu$ is denoted $\rho_\mu$ \znote{Introduced the notation $\rho_\mu$. I think this can be useful in the technical sections as well.} and defined as:
\begin{equation}\label{eq:rhodist}
\rho_\mu = \bbE_{\ket{\psi} \gets \mu}\Big[ \big( \ket{\psi}\bra{\psi} \big) \Big] :=
\int_{\ket{\psi} \in {\cal S}(2^n)} \big( \ket{\psi}\bra{\psi} \big) \text{d}\mu(\ket{\psi}) \enspace .
\end{equation}

\paragraph{Statistical Distance.}
We use basic properties of the statistical distance metric (also known as total variation distance).
Statistical distance can be described in terms of operations, that is, for two (possibly continuous) distributions $D_1$, $D_2$ with corresponding supports $S_1$, $S_2$, the statistical distance between $D_1, D_2$ is the maximal advantage,
$$
\abs{ \Pr_{x \gets D_1}\left[ \A(x) = 1 \right] - \Pr_{x \gets D_2}\left[ \A(x) = 1 \right] }
$$
taken over all functions $\A : S_1 \cup S_2 \rightarrow \{ 0, 1 \}$. We note that we can allow $\A$ to be randomized and obtain an equivalent definition. The statistical distance between two random variables is the statistical distance between their associated distributions.

Additionally, throughout the proof of Theorem \ref{thm:main} we will use the following fact about the statistical distance between a distribution and a conditional version of it.


\znote{Then maybe it is worthwhile to formally state it as fact and refer to it when needed?} \omri{Done.}\znote{This looks over-complicated. What you are stating is a very simple fact which is very simple to state in random variable terminology. How about the following.}\omri{Yeah, looks better, thanks.}
\begin{fact} \label{fact:statistical_condition}
	Let $X$ be a random variable and $E$ some probabilistic event. Denote $Y = X | \bar{E}$, i.e.\ the conditional variable of $X$ conditioned on $E$ \emph{not} happening. Then
	$$
	\sd(X, Y) \le \Pr[E] \enspace .
	$$
\end{fact}

\paragraph{Trace Distance.}
The trace distance, defined below, is a generalization of statistical distance to the quantum setting and represents the maximal distinguishing probability between distributions over quantum states.

\znote{Why in the classical case you define SD using distinguishing advantage but in the quantum case you define TD using $L_1$ norm? Isn't it better to define both in the same way?}\omri{Right, the norm-based definition is not really used in this paper.}


\begin{definition} [Trace Distance]
    Let $\rho_0, \rho_1 \in {\cal D}(2^n)$ be two density matrices of $n$-qubit mixed states.
	For a projective measurement $\A$ with output in $\{ 0,1 \}$ define
	$$
	\Delta_{\A, \rho_0, \rho_1} := \abs{
		\Pr\Big[ \A\big( \rho_0 \big) = 0 \Big] - 
		\Pr\Big[ \A\big( \rho_1 \big) = 0 \Big]
	} \enspace .
	$$
	The trace distance between $\rho_0$, $\rho_1$ is
	$$
	\td( \rho_0, \rho_1 ) := \max_{\{ 0,1 \}\text{ projective measurement }\A} \Delta_{\A, \rho_0, \rho_1} \enspace.
	$$
\end{definition}
We note that the trace distance is often equivalently defined as $\tfrac{1}{2} \norm{\rho_0 -\rho_1}_1$, where $\norm{\cdot}_1$ refers to the $\ell_1$ norm of the vector of eigenvalues of the operand matrix.

A standard fact about trace distance is the following.
\begin{fact}
    Let $D_0$, $D_1$ be two distributions over $n$-qubit states and let $\rho_0, \rho_1 \in {\cal D}(2^n)$ be the corresponding density matrices.
    For a projective measurement $\A$ with output in $\{ -1,1 \}$ define
	$$
	\Tilde{\Delta}_{\A, \rho_0, \rho_1} := \abs{
		\bbE_{\substack{\ket{\psi} \gets D_0, \\ \text{Measurement}}}
		\Big[ \A\big( \ket{\psi} \big) \Big] - 
		\bbE_{\substack{\ket{\psi} \gets D_1, \\ \text{Measurement}}}
		\Big[ \A\big( \ket{\psi} \big) \Big]
	} \enspace .
	$$
	Then,
	$$
	2\cdot \td( \rho_0, \rho_1 ) = \max_{\{ -1,1 \}\text{ projective measurement }\A} \Tilde{\Delta}_{\A, \rho_0, \rho_1} \enspace.
	$$
\end{fact}



The trace distance between pure states is given by the following expression.
\begin{fact}
    For $n$-qubit pure quantum states $\ket{\psi}, \ket{\phi}$, the trace distance between them is:
    $$
    \td\Big( \ketbra{\psi}{\psi}, \ketbra{\phi}{\phi} \Big) = \sqrt{1 - \abs{\braket{\psi}{\phi}}^2} \enspace .
    $$
\end{fact}

Trace distance is an operator on density matrices.
In this work we will sometimes use it directly on distributions, that is we denote $\td(D_1, D_2)$, where $D_1, D_2$ are distributions over $n$-qubit quantum states. This notation refers to the trace distance between the two density matrices induced by $D_1$ and $D_2$ (as per Eq.~\eqref{eq:rhodist}). That is,
$$
\td(D_1, D_2):= \td(\rho_{D_1}, \rho_{D_2}) =
\td
\bigg( 
\bbE_{\ket{\psi} \gets D_1} \Big[ \big( \ket{\psi}\bra{\psi} \big) \Big] \; , \; 
\bbE_{\ket{\psi} \gets D_2} \Big[ \big( \ket{\psi}\bra{\psi} \big) \Big]
\bigg) \enspace .
$$

\paragraph{Quantum Unitary for a Classical Function.}
Let $f : \{ 0, 1 \}^n \rightarrow \{ 0, 1 \}^m$ be a function. The unitary of $f$ is denoted by $U_f$, it is a unitary over $n + m$ qubits defined as
$$
\forall x \in \{ 0, 1 \}^n, y \in \{ 0, 1 \}^m : U_f\ket{x, y} := \ket{x, y \oplus f(x)} \enspace .
$$

\paragraph{Quantum Rejection Sampling.} \label{par:qrs}
Quantum Rejection Sampling (QRS) is a known efficient procedure for taking one quantum state $\ket{\alpha}$ and outputting with some probability a different quantum state $\ket{\beta}$, given black box access to a circuit that describes their closeness.
Formally, the algorithm $\qrs$ gets as input an $n$-qubit quantum state $\ket{\alpha}$ and quantum oracle access to a unitary $U$ on $n + k$ qubits (where $k$ is related to the binary description length for complex numbers that is being used) and have the following correctness and time complexity guarantees.

\begin{theorem} [Quantum Rejection Sampling] \label{thm:qrs}
    Let $\ket{\alpha}, \ket{\beta}$ be two $n$-qubit quantum states and let $U$ be an $(n + k)$-qubit unitary.
    Assume there exists a positive real number $d$ such that the following hold
    \begin{itemize}
	    \item
	    $d \geq \max_{x \in \{ 0, 1 \}n} \abs{\frac{\beta_x}{\alpha_x}}$.
	    
	    \item
	    $\forall x \in \{ 0, 1 \}^n$, the complex number $\frac{\left( \beta_x / \alpha_x \right)}{d}$ can be described with full precision in $k$ bits.
	
	    \item
	    $U$ is the unitary of the classical function $f : \{ 0, 1 \}^n \rightarrow \{ 0, 1 \}^k$ such that $f(x) :=  \frac{ \left( \beta_x / \alpha_x \right) }{d}$.
    \end{itemize}
    Then $\qrs^{U}(\ket{\alpha})$ outputs $(\success, \ket{\beta})$ with probability at least $\frac{1}{d^2}$ and otherwise outputs $(\fail, \ket{0^n})$.

    The algorithm makes a single query to $U$, and assuming this query takes a single time step, the time complexity of $\qrs^{U}(\ket{\alpha})$ is $\poly(n, k)$.
\end{theorem}

\subsection{Pseudorandom Functions and $m$-Wise Independent Functions}

We define pseudorandom functions with quantum security (QPRFs).	
\begin{definition} [Quantum-Secure Pseudorandom Function (QPRF)] \label{def:qprf}
	Let $\Keys = \{ \Keys_n \}_{n \in \bbN}$ be an efficiently samplable key distribution, and let $\PRF = \{ \PRF_n \}_{n \in \bbN}$, $\PRF_n : \Keys_n \times \binset^n \rightarrow \binset^{\poly(n)}$ be an efficiently computable function, where $\poly(\cdot)$ is some polynomial.
	We say that $\PRF$ is a quantum-secure pseudorandom function if for every efficient non-uniform quantum algorithm $A = \{ A_n \}_{n \in \bbN}$ (with quantum advice) that can make quantum queries there exists a negligible function $\negl(\cdot)$ s.t. for every $n \in \bbN$,
	$$
	\abs{\Pr_{k \gets \Keys_n}[A_n^{\PRF_k } = 1] - \Pr_{f \gets (\binset^n)^{(\binset^n)}}[A_n^{f} = 1]} \leq \negl(n) \enspace .
	$$
\end{definition}
In \cite{zhandry2012construct}, QPRFs were proved to exist under the assumption that post-quantum one-way functions exist. 

We define $m$-wise independent functions as keyed functions s.t.\ when the key is sampled from the key distribution, then any $m$ different inputs to the function generate $m$-wise independent random variables.

\begin{definition}[$m$-Wise Independent Function] \label{def:t-wise_function}
	Let $n, m, p \in \bbN$, let $\Keys$ be a key distribution, and let $f$, $f : \Keys \times \binset^n \rightarrow \binset^{p}$ a function.
	$(f, \Keys)$ is an $m$-wise independent function if for every distinct $m$ input values $x_1, \cdots, x_{m} \in \binset^{n}$,
	$$
	\forall y_1, \cdots, y_{m} \in \binset^{p} : 
	\Pr_{k \gets \Keys}[f(k, x_1) = y_1 \land \cdots \land f(k, x_{m}) = y_{m}] = 2^{-p \cdot m} \enspace .
	$$
\end{definition}

Based on $m$-wise independent functions we define efficiently samplable $m$-wise independent function families.

\begin{definition}[Efficient $m(n)$-Wise Independent Function] \label{def:t-wise}
	Let $m(n), p(n) : \bbN \rightarrow \bbN$ be functions, let $\Keys = \{ \Keys_n \}_{n \in \bbN}$ be an efficiently samplable key distribution, and let $f = \{ f_n \}_{n \in \bbN}$, $f_n : \Keys_n \times \binset^n \rightarrow \binset^{p(n)}$ be an efficiently computable function.
	Then, if for every $n \in \bbN$, $(f_n, \Keys_n)$ is an $m(n)$-independent function, then $(f, \Keys)$ is an efficient $m(n)$-wise independent function.
\end{definition}


\subsection{Quantum Randomness and Pseudorandomness}

\subsubsection{The Haar Measure}
The Haar measure on quantum states is the quantum analogue of the classical uniform distribution over classical bit strings. That is, it is the uniform (continuous) probability distribution on quantum states.
Recall that an $n$-qubit quantum state can be viewed as a unit vector in $\bbC^{2^n}$, thus the Haar measure on $n$ qubits is the uniform distribution over all unit vectors in $\bbC^{2^n}$. In this work we denote the $n$-qubit Haar distribution with $\mu_n$.
From this point forward we refer to the uniform distribution over quantum states simply as ``random'', and don't mention specifically that it is  with respect to the Haar distribution.


\subsubsection{Scalable Asymptotically Random State Generators}
We propose a scalable variant to the notion of Asymptotically Random State (ARS) generators which was implicitly defined in \cite{JLS18} and explicitly in \cite{BS19}.
Previous works consider an ARS generator to be an efficient quantum algorithm $\Gen$ that gets quantum oracle access to $U_f : \ket{x, y} \rightarrow \ket{x, y \oplus f(x)}$ for a random classical function $f$, along with a parameter $n \in \bbN$ that denotes the number of desired output qubits.
The guarantee of the ARS generator is that for any polynomial $t(n)$ in $n$, $t(n)$ outputs from $\Gen^{U_f}$ (executed with the same function $f$) have negligible trace distance (\emph{in $n$}) from $t(n)$-copies of a random $n$-qubit state
This means that $n$ plays two roles, it denotes the number of qubits in the output state but also the security parameter that determines the quality of randomness (i.e. how indistinguishable it is from random).

A {\it Scalable} ARS generator is one that gets two parameters $n, \secp$ instead of one.
$n$, as before, denotes the number of wanted output qubits, and $\secp$ is a security parameter, thus a scalable ARS generator eliminates the dependence between state size and security.

\begin{definition} [Asymptotically Random State (ARS) Generator] \label{def:ars}
	A quantum polynomial-time algorithm $\Gen$ with input $(1^n, 1^\secp)$ for $n, \secp \in \bbN$ and quantum oracle access to $U_f : \ket{x, y} \rightarrow \ket{x, y \oplus f(x)}$ for $f : \{ 0, 1 \}^n \rightarrow \{ 0, 1 \}^{\poly(n, \secp)}$, is an ARS generator if there exists a negligible function $\negl(\cdot)$ s.t. for every polynomial $t : \bbN \rightarrow \bbN$, for all natural numbers $n, \secp$,
	$$
	\td
	\big( 
	D_1 , D_2 
	\big) \leq
	\negl(\secp) \enspace ,
	$$
	where the distributions $D_1, D_2$ are defined as follows.
	\begin{itemize}
		\item $D_1 : $
		Sample $f \gets \left( \{ 0, 1 \}^{\poly(n, \secp)} \right)^{ \{ 0, 1 \}^n }$, perform $t(\secp)$ independent executions of $\Gen^{U_f}(1^n, 1^\secp)$ and output the $t(\secp)$ output quantum states.
		
		\item $D_2 : $
		Sample $\ket{\psi} \gets \mu_n$ a random $n$-qubit quantum state, and output $t(\secp)$ copies of it: $\ket{\psi}^{\otimes t(\secp)}$. Recall that $\mu_n$ is the Haar measure on $n$ qubits.
	\end{itemize}
\end{definition}


We next define (scalable) quantum state $t$-design generators and (scalable) pseudorandom quantum state (PRS) generators.
After defining these, we briefly describe a general and simple reduction structure that shows how to construct $t$-designs and PRS generators from any ARS generator.

\subsubsection{Approximate Quantum State $t$-Designs}
A quantum state $t$-design \cite{ambainis2007quantum} is a distribution over quantum states that mimics the uniform distribution over quantum states when the number of output copies is restricted to $t$. 
A (scalable, approximate) quantum state $t$-design generator consists of two quantum algorithms $K, G$.
The key sampler algorithm $K$ samples a classical key $k$ given two parameters $1^n, 1^\secp$ where $n$ denotes the number of qubits and $\secp$ denotes the security parameter.
The state generation algorithm $G$ gets a key $k$ and outputs an $n$-qubit state $\ket{\psi}$.
Informally, the randomness gaurantee of a $t$-design generator is that if we sample a key $k$ once from $K(1^n, 1^\secp)$ and then execute $G(k)$ $t$ times and output the $t$ outputs, then this output distribution is going to be indistinguishable from $t$ copies of an $n$-qubit quantum state, for unbounded quantum distinguishers.
The formal definition follows.

\begin{definition} [$\varepsilon(\secp)$-Approximate State $t(\secp)$-Design Generator]
	Let $\varepsilon(\secp): \bbN \rightarrow [0, 1]$, $t(\secp): \bbN \rightarrow \bbN$ be functions.
	We say that a pair of quantum algorithms $(K, G)$ is an $\varepsilon(\secp)$-approximate state $t(\secp)$-design generator if the following holds:
	\begin{itemize}
		\item \textbf{Key Generation.}
		For all $n, \secp \in \bbN$, $K(1^n, 1^\secp)$ always outputs a classical key $k$.
		
		\item \textbf{State Generation.}
		Given $k$ in the support of $K(1^n, 1^\secp)$ the algorithm $G(1^n, 1^\secp, k)$ will always output an $n$-qubit quantum state.
		
		\item \textbf{Approximate Quantum Randomness.}
		For all $n, \secp \in \bbN$, 
		$$
		\td
		\big( 
		D_1 , D_2 
		\big) \leq
		\varepsilon(\secp) \enspace ,
		$$
		where the distributions $D_1, D_2$ are defined as follows.
		\begin{itemize}
			\item $D_1 : $
			Sample $k \gets K(1^n, 1^\secp)$, perform $t(\secp)$ independent executions of $G(1^n, 1^\secp, k)$ and output the $t(\secp)$ output quantum states.
			
			\item $D_2 : $
			Sample $\ket{\psi} \gets \mu_n$ a random $n$-qubit quantum state, and output $t(\secp)$ copies of it: $\ket{\psi}^{\otimes t(\secp)}$.
		\end{itemize}
	\end{itemize}
\end{definition}
It is not part of the standard definition, but it is usually the case that the algorithms $K, G$ execute in time $\poly(n, \secp)$, which is going to be the case in this work as well.

\subsubsection{Pseudorandom Quantum States}
We define scalable Pseudorandom State (PRS) generators.
Compared to $t$-designs, Quantum Pseudorandom State Generators have a slight difference, and formally incomparable randomness guarantee.
Mainly, with a PRS we are guaranteed that the output state is going to be indistinguishable for any polynomial number of copies $t(\secp)$ {\it without} knowing in advance $t(\secp)$, however this indistinguishability is only {\it computational}. That is, it is only guaranteed that computationally bounded distinguishers will be unable to tell the difference between $t(\secp)$ executions of the generator and $t(\secp)$ copies of a random quantum state.
The scalability property maintains the ability to increase security without increasing the state size $n$.
We remind that the notion of scalability in PRS generators was not considered in previous works \cite{JLS18, BS19} and thus the following definition differs a bit from the previous definition of a PRS, we view this as the more proper definition.

\begin{definition} [Scalable Pseudorandom Quantum State (PRS) Generator] \label{def:prs}
	We say that a pair of polynomial-time quantum algorithms $(K, G)$ is a Pseudorandom State (PRS) Generator if the following holds:
	\begin{itemize}
		\item \textbf{Key Generation.}
		For all $n, \secp \in \bbN$, $K(1^n, 1^\secp)$ always outputs a classical key $k$.
		
		\item \textbf{State Generation.}
		Given $k$ in the support of $K(1^n, 1^\secp)$ the algorithm $G(1^n, 1^\secp, k)$ will always output an $n$-qubit quantum state.
		
		\item \textbf{Quantum Pseudorandomness.}
		For any polynomial $t(\cdot)$ and a non-uniform polynomial-time quantum algorithm $A = \{ A_\secp \}_{\secp \in \bbN}$ (with quantum advice) there exists a negligible function $\negl(\cdot)$ such that for all $n, \secp \in \bbN$,
		$$
		\abs{\Pr[A_\secp \big( D_1 \big) = 1] -
			\Pr[A_\secp \big( D_2 \big) = 1]} \leq \negl(\secp) \enspace ,
		$$
		where the distributions $D_1, D_2$ are defined as follows.
		\begin{itemize}
			\item $D_1 : $
			Sample $k \gets K(1^n, 1^\secp)$, perform $t(\secp)$ independent executions of $G(1^n, 1^\secp, k)$ and output the $t(\secp)$ output quantum states.
			
			\item $D_2 : $
			Sample $\ket{\psi} \gets \mu_n$ a random $n$-qubit quantum state, and output $t(\secp)$ copies of it: $\ket{\psi}^{\otimes t(\secp)}$.
		\end{itemize}
	\end{itemize}
\end{definition}

\subsubsection{Scalable PRS and Quantum State $t$-Design Generators from Scalable ARS Generators}
\label{sec:consequences}

We recall a generic transformation from previous works that explain how to construct PRS generators and quantum state $t$-designs from any ARS generator.
We start with the paradigm from \cite{JLS18, BS19} that explains a simple way to turn any ARS generator into a PRS generator.

\begin{lemma}
	If there exists a scalable ARS generator and post-quantum one-way functions exist, then there exists a scalable PRS generator.
\end{lemma}

\begin{proof}[Proof Sketch.]
The proof follows the same lines as the proof of \cite[Claim 4, Section 3.1]{BS19}, with the additional scalability property.
The key generator $K(1^n, 1^\secp)$ of the PRS is the key generator of some quantum-secure pseudorandom function $\PRF$ with security parameter $n + \secp$.
For a sampled PRF key $k$, the state generator algorithm $G$ simply executes the ARS generator with the pseudorandom function instead of the truly random function, $G(1^n, 1^\secp, k) := \Gen^{U_{\PRF_k}}(1^n, 1^\secp)$.
For a polynomial $t(\cdot)$, $t(\secp)$ copies of the generated distribution are computationally indistinguishable (by quantum adversaries) from $t(\secp)$ copies of the standard output distribution of the ARS generator, by the security guarantee of the PRF.
Additionally, $t(\secp)$ copies of the output distribution of the ARS is already known to be indistinguishable (by unbounded distinguishers) from $t(\secp)$ copies of a random quantum state, and our proof is concluded.
\end{proof}

Also, we follow the observation from \cite{BS19} that explains how an ARS generator implies the existence of $t$-designs (with depth that has logarithmic dependence on $t$).

\begin{lemma}
	Assume there exists a scalable ARS generator with the following properties:
	\begin{itemize}
		\item The generator is implemented by a circuit of depth $T(n, \secp)$.
		\item For all $n, \secp, t$ its output is $\varepsilon(n, \secp, t)$-indistinguishable from a $t$-tensor of a random $n$-qubit state.
	\end{itemize}
	 Then there exists an $\varepsilon(n, \secp, t)$-approximate scalable $t$-design generator, which is implementable by circuits of depth
	$$
	T(n, \secp)\cdot \log(n) \cdot \log (2\cdot t \cdot T(n, \secp)) \enspace .
	$$
\end{lemma}

\begin{proof}[Proof Sketch.]
The proof is similar to the explanation in \cite[Section 3.2]{BS19}, with slight differences and an additional consideration of the scalability property.
The key generator $K(1^n, 1^\secp)$ of the $t$-design samples an efficient $m$-wise independent function $\tilde{f}$, where $m := 2t\cdot T(n, \secp)$. The state generator algorithm $G$ executes the ARS generator with the function $\tilde{f}$ instead with the truly random function, $G(1^n, 1^\secp, \tilde{f}) := \Gen^{U_{\tilde{f}}}(1^n, 1^\secp)$.
By \cite[Fact 2]{zhandry2012construct}, The behavior of any quantum algorithm making at most $m$ quantum queries to a $2m$-wise independent function is identical to its behavior when the queries are made to a random function. Therefore if we make $t$ executions of $G(1^n, 1^\secp, \tilde{f})$, each of which makes at most $T(n, \secp)$ queries to $U_{\tilde{f}}$, then the output distribution of the algorithm $G(1^n, 1^\secp, \tilde{f})$ is the same as that produced by the ARS generator (when it uses a truly random function).
Since the classical depth of an $m$-wise independent function on $n$ bits is $\log(n) \cdot \log (m)$, the proof follows (see elaboration on the classical depth of $m$-wise independent functions in \cite[Section 3.2]{BS19}).
\end{proof}

\subsection{The Continuous Gaussian and Rounded Gaussian Distributions}

In this work we will work with distributions related to the Gaussian distribution over $\bbR$ denoted $\normalDist$, also known as the normal distribution having a mean of 0 and variance of 1.
More specifically we will consider the complex Gaussian distribution over $\bbC$, denoted $\normalCompDist$, where both real and imaginary parts of a complex number are sampled independently from $\normalDist$.

\paragraph{Rounded Gaussian Distribution.}
The true Gaussian distribution is continuous and we cannot {\it exactly} sample from it. Instead, we will use a discrete distribution that we can efficiently sample from.
There are quite a few versions of distributions that are discretezations of the Gaussian distribution. In this work we use the rounded Gaussian distribution, which we denote by $\SingleDistribution{\varepsilon, B}$. This distribution is parameterized by $\varepsilon = 2^{-m} > 0$ (for some $m \in \bbN$) and by $B \in \bbN$, where $B$ is some integer multiple of $\varepsilon$.


To define the distribution $\SingleDistribution{\varepsilon, B}$ we first define the rounding function $R_{(\varepsilon, B)}(\cdot)$.
For a number $x \in \bbR$, if $|x| > B$ then $R_{(\varepsilon, B)}(x) := 0$, and otherwise $R_{(\varepsilon, B)}(x)$ rounds $x$ up (in absolute value) to the nearest multiple of $\varepsilon$.
Formally, if $|x| \leq B$ then $R_{(\varepsilon, B)}(x)$ is the number $y \in \bbR$ that has minimal absolute value and s.t. both $|x| \leq |y|$, $\exists k \in \bbZ : y = k\cdot \varepsilon$.
For a complex number $z \in \bbC$, $R_{(\varepsilon, B)}(z)$ is just applying $R_{(\varepsilon, B)}(\cdot)$ to both real and imaginary parts of $z$.

We define $\SingleDistribution{\varepsilon, B}$ to be the output distribution of the following process: Sample $z \gets \normalCompDist$ and output $R_{(\varepsilon, B)}(z)$.
The output of $\SingleDistribution{\varepsilon, B}$ is specified by a number between $0$ and $B$ with precision $\varepsilon = 2^{-m}$, thus the output length in bits is bounded by $m + \ceil{\log_2(B)}$.

We use the following standard fact about (classical) Gaussian sampling.
\begin{fact} [Efficient Rounded Gaussian Sampling] \label{fact:gsample}
    There is a sampling algorithm $\gauss{}$ that takes $1^m, B$ (and random tape) as input, runs in polynomial time, i.e.\ $\poly(m, \log B)$, and samples from a distribution that has statistical distance at most $2^{-m}$ from the rounded Gaussian distribution $\SingleDistribution{2^{-m}, B}$.
\end{fact}

\section{General Tools for Quantum Information}
\label{sec:qit_tools}

\subsection{State Generation of Balanced Vectors}
In this subsection we describe a simple procedure that given quantum oracle access to the entries of some general, not necessarily normalized vector $v \in \bbC^{2^n}$, generates the $n$-qubit quantum state $\ket{v}$ that corresponds to (the normalization of) $v$.
More formally, the procedure gets two pieces of information about $v$:
\begin{itemize}
	\item 
	Quantum oracle access to $U_v$, the unitary of the classical function $v : \{ 0, 1 \}^n \rightarrow \{ 0, 1 \}^k$ (where $k$ is the description size in bits of each entry of $v$) that describes the vector $v$ and maps $v(x) := v_x$.
	
	\item 
	An upper bound $M \in \bbN$ on any entry of $v$, that is, $\max_{x \in [N]}|v_x| \leq M$.
\end{itemize}
The procedure runs in time $\poly(n, k, \log M)$ and outputs the quantum state $\ket{v}$ as follows.

\paragraph{$\bvs^{U_v} \left( M \right):$} \znote{Maybe use some LaTeX environment in order to introduce}\omri{Not sure what do you mean here.}
\begin{enumerate}	
	\item
	Define the quantum unitary $U_{\tilde{v}}$ which is the unitary of the classical function $\tilde{v} : \{ 0, 1 \}^n \rightarrow \{ 0, 1 \}^{k + \log(M)}$ that maps $\tilde{v}(x) := v_x\cdot \frac{1}{M}$.
	
	It's trivial to simulate $U_{\tilde{v}}$ given $U_v$: Given a query $\ket{x, y}$, we concatenate an ancilla of zeros and apply $U_v$ to get $\ket{x, y, v_x}$, then apply a simple unitary that multiplies by $\frac{1}{M}$ (on the last register as input and on the second register as output) to obtain $\ket{x, y \oplus v_x \cdot \frac{1}{M}, v_x}$, and then use $U_v$ again to uncompute the last register.
	
	\item 
	Execute quantum rejection sampling, $\left( b \, , \, \ket{\hat{v}} \right) \gets \qrs^{U_{\tilde{v}}}\left( \ket{+}^{\otimes n} \right)$ (see specification of $\qrs$ in Theorem~\ref{thm:qrs}).
	If $b = \fail$ then output $\fail$, otherwise output $\ket{\hat{v}}$. \znote{What is QRS? If it was defined before provide reference to the definition.}\omri{done.}\znote{There reference is unclear. You cannot just refer to ``3.1'' without saying what it is. Also why are you referring to the section and not to the theorem statement?}\omri{Now?}
\end{enumerate}

\begin{cclaim}[$\bvs$ Success Probability]
	If $\max_{i \in [N]}|v_i| \leq M$ then the execution $\bvs^{U_v}(M)$ always outputs either $\fail$ or the quantum state $v \cdot \frac{1}{\norm{v}} = \ket{\hat{v}}$, furthermore the execution succeeds and outputs the quantum state $\ket{\hat{v}}$ with probability at least $\frac{\norm{v}^2}{M^2\cdot N}$. 
\end{cclaim}

\begin{proof}
	We need to make sure that we execute the quantum rejection sampling algorithm $\qrs$ with correct parameters (specified in Theorem~\ref{thm:qrs}), and also understand what exactly are the parameters for $\qrs$.
	As the starting state $\ket{\alpha}$ we input $\ket{+}^{\otimes n}$, our target state $\ket{\beta}$ is $\ket{\hat{v}}=\tfrac{\ket{v}}{\norm{v}}$.
	As the state transformation unitary $U$ we use $U_{\tilde{v}}$, that is, the unitary of the classical function $f(x) := v_x \cdot \frac{1}{M}$.

	It follows that there exists an upper bound $d \geq \max_{x \in \{ 0, 1 \}n} \abs{\frac{\beta_x}{\alpha_x}}$ s.t. $\forall x \in \{ 0, 1 \}^n : f(x) :=  f(x) := v_x \cdot \frac{1}{M} = \frac{\beta_x / \alpha_x}{d}$, by taking $d := \frac{M \cdot \sqrt{N}}{\norm{v}}$.
	\begin{itemize}
		\item $d$ is indeed an upper bound:
		$$
		\forall x \in \{ 0, 1 \}^n :
		\abs{\frac{\beta_x}{\alpha_x}} =
		\abs{\frac{v_x / \norm{v}}{1 / \sqrt{N}}} = 
		\abs{v_x \cdot \frac{\sqrt{N}}{\norm{v}}} \leq 
		\frac{M \cdot \sqrt{N}}{\norm{v}} \enspace .
		$$
		
		\item $f(\cdot)$ indeed computes $\frac{\beta_x}{\alpha_x}/d$:
		\begin{align*}
		\forall x \in \{ 0, 1 \}^n :
		f(x)
		&:= v_x \cdot \frac{1}{M} \\
		&= v_x \cdot \frac{1}{M} \cdot \frac{\sqrt{N}}{\sqrt{N}} \cdot \frac{\norm{v}}{\norm{v}} \\
		&= \left( \frac{v_x}{\norm{v}} \cdot \sqrt{N} \right) \cdot \left( \frac{1}{M} \cdot \frac{\norm{v}}{\sqrt{N}} \right) \\
		&= \frac{\left( \frac{v_x / \norm{v}}{1 / \sqrt{N}} \right)}{\left( \frac{M \cdot \sqrt{N}}{\norm{v}} \right)} \\
		&= \frac{\beta_x / \alpha_x}{d} \enspace .
		\end{align*}
	\end{itemize}
	
	The conditions for $\qrs$ hold, \znote{Refer to rejection sampling lemma.}\omri{done.}\znote{Did you ever see a paper that refers to ``3.5'' without saying what it is?}\omri{Thinking about it again, I moved the QRS reference to the beginning of the proof. Is it ok?} and thus from the correctness guarantee of quantum rejection sampling we can be sure that the algorithm $\bvs$ will always output either $\fail$ or $\ket{\beta} := \ket{\hat{v}}$.
	As for the probability of success in outputting $\ket{\hat{v}}$, again from the success guarantees of $\qrs$ this probability is at least $\frac{1}{d^2} = \frac{\norm{v}^2}{M^2\cdot N}$.
\end{proof}

The above procedure tries to generate $\ket{\hat{v}}$ once and it will be convenient to have an amplified version of this algorithm as a black box, this is an option because we can always re-generate the state $\ket{+}^{\otimes n}$ efficiently and retry.
The amplified version of the algorithm is with the same name and have one more parameter $k \in \bbN$ (amplification parameter), that is, $\bvs^{U_v}(M, k)$.

The amplified version of $\bvs$ executes $k$ (parallel) repetitions of $\bvs^{U_{v}}(M)$, if all fail it outputs $\fail$, and if either succeeds it outputs the generated state $\ket{\hat{v}}$.
The probability of $\bvs^{U_v}(M, k)$ to succeed in generating the state $\ket{\hat{v}}$ follows.

\begin{cclaim}[Amplified $\bvs$ Success Probability] \label{claim:amp_bvs_success}
	If $\max_{i \in [N]}|v_i| \leq M$ then the algorithm $\bvs^{U_v}(M, k)$ always outputs either $\fail$ or the quantum state $v \cdot \frac{1}{\norm{v}} = \ket{\hat{v}}$, furthermore the algorithm succeeds and outputs the quantum state $\ket{\hat{v}}$ with probability at least $1 - e^{-\frac{k \cdot \norm{v}^2}{M^2\cdot N}}$.
\end{cclaim}

\begin{proof}
	\begin{align*}
	\Pr\left[ \bvs^{U_v}(M, k) = \fail \right] 
	&= \Pr\left[ \bvs^{U_v}(M) \text{ failed $k$ times in a row} \right] \\
	&= \left( \Pr\left[ \bvs^{U_v}(M) = \fail \right] \right)^{k} \\
	&\leq \left( 1 - \frac{\norm{v}^2}{M^2 \cdot N} \right)^{k} \\
	&\leq e^{-\frac{k \cdot \norm{v}^2}{M^2\cdot N}} \enspace .
	\end{align*}
\end{proof}

\subsection{Analytic Tools for Distributions}
In this subsection we describe some analytic tools for bounding trace norm between two distributions, for multiple output copies.
We start with an elementary property of trace distance and classical statistical distance that we will use in our construction.

\begin{lemma}[Classical Statistical Distance Implies Trace Distance] \label{lemma:sd_td}
	Let $n \in \bbN$ and let $D_1$, $D_2$ be two distributions over unit vectors in $\bbC^{2^n}$.
	Let $\tilde{D_1}$, $\tilde{D_2}$ be the quantum-state distributions of $D_1$, $D_2$, that is, for $b \in \{ 0, 1 \}$, a sample from $\tilde{D_b}$ is generated by sampling a vector $v$ from $D_b$, and outputting an $n$-qubit register in the state described by $v$.
	
	Then, if $\sd(D_1, D_2) \leq \varepsilon$, then for every number of copies $t \in \bbN$, 
	$$
	\td\left(
	\bbE_{\ket{v} \gets \tilde{D_1}}\left[ \left( \ket{v}\bra{v} \right)^{\otimes t} \right],
	\bbE_{\ket{v} \gets \tilde{D_2}}\left[ \left( \ket{v}\bra{v} \right)^{\otimes t} \right]
	\right) \leq \varepsilon \enspace .
	$$
\end{lemma}
\begin{proof}
	Intuitively, the proof follows from the fact that a computationally unbounded mapping can always capture the computation of an (even unbounded) quantum process, along with the fact that when the classical description of a state is available then there is no advantage in having more than a single copy.
	Formally, we assume towards contradiction there is a projective measurement $\A$ (with output in $\{ 0, 1 \}$) that distinguishes between a $t$-tensor of $\tilde{D_1}$ and a $t$-tensor of $\tilde{D_2}$ with advantage bigger than $\varepsilon$, and describe a (randomized) distinguisher $\A' : \bbC^{2^n} \rightarrow \{ 0, 1 \}$ that distinguishes between $D_1, D_2$ with advantage bigger than $\varepsilon$. Let $A$ denote the Hermitian matrix that corresponds to the projective measurement $\A$.
	
	The distinguisher $\A'$ is defined as follows. Given an input $v \in \bbC^{2^n}$, consider the vector $\ket{v'} = \ket{v}^{\otimes t}$, and compute the value $p = \bra{v'} A \ket{v'}$. We note that this value is exactly the probability that $\A$ outputs $1$ when input the quantum state $\ket{v}^{\otimes t}$. The distinguisher $\A'$ then outputs $1$ with probability $p$ and $0$ with probability $1-p$. By definition, the advantage of $\A'$ in distinguishing $D_1$ and $D_2$ is identical to the advantage of $\A$ in distinguishing the $t$-tensored $\tilde{D_1}$ and $\tilde{D_2}$. \znote{I think the previous version was not well defined. I rewrote. Please check that I didn't make mistakes, I hope I remember the properties of projective measurements correctly.}
	\omri{I see. The above looks good to me, thanks.}
%
%
\end{proof}

\paragraph{Robustness to Small Shifts.} Lemma~\ref{lemma:sd_td} asserts that distributions on quantum states are indistinguishable if they are induced by indistinguishable distributions over vectors in the respective Hilbert space. This is a very strict condition and in fact in many cases distributions on quantum states can be indistinguishable even if the the respective distributions over vectors are highly distinguishable. 

This will be useful in the context of this work since we wish to show indistinguishability between the Haar random distribution, which corresponds to a continuous distribution over the sphere, and an efficiently samplable distribution (with oracle access to a random function), which necessarily produces a discrete distribution over vectors. Hence, the two distributions over vectors are necessarily distinguishable (with advantage $1$), and yet we will be able to bound the distinguishing gap between the quantum states.

Technically, we rely on the well known property that quantum states that correspond to vectors with inner product close to $1$ are indistinguishable. This is formalized in Lemma \ref{lemma:angular} below, which considers a distribution over vectors, and a small perturbation of this distribution, that does not shift the vector by too much. We show that such perturbation, which in particular captures the case of rounding a continuous distribution into some discrete domain, would be indistinguishable in terms of the resulting quantum state.

We start with the following simple auxiliary lemma.

\begin{lemma}[Trace Distance and Diameter of Supports Union] \label{claim:trace_diameter}
	Let $\varepsilon \in [0,1]$, and let $D_1, D_2$ be two distributions over unit vectors in $\bbC^{2^n}$ and denote their respective supports by $S_1$, $S_2$.
	Assume that for every pair $\ket{u} \in S_1, \ket{v} \in S_2$ we have $\abs{\braket{u}{v}} \geq 1 - \varepsilon$, then for all $t \in \bbN$,
	$$
	\td\Big(
	\bbE_{\ket{v} \gets D_1}\left[ \left( \ket{v}\bra{v} \right)^{\otimes t} \right],
	\bbE_{\ket{v} \gets D_2}\left[ \left( \ket{v}\bra{v} \right)^{\otimes t} \right]
	\Big) \leq \sqrt{1 - \left( 1 - \varepsilon \right)^{2\cdot t}} \enspace .
	$$
\end{lemma}

\begin{proof}
	By averaging,
	$$
	\td\Big(
	\bbE_{\ket{v} \gets D_1}\left[ \left( \ket{v}\bra{v} \right)^{\otimes t} \right],
	\bbE_{\ket{v} \gets D_2}\left[ \left( \ket{v}\bra{v} \right)^{\otimes t} \right]
	\Big) \leq
	\max_{\ket{u} \in S_1, \ket{v} \in S_2}
	\td\Big(
	\left( \ket{u}\bra{u} \right)^{\otimes t},
	\left( \ket{v}\bra{v} \right)^{\otimes t}
	\Big) \enspace ,
	$$
	and due to the simple trace distance formula for pure states, the last expression is exaclty the following:
	$$
	\sqrt{1 - \abs{\braket{u^{\otimes t}}{v^{\otimes t}}}^2} =
	\sqrt{1 - \abs{\braket{u}{v}}^{2\cdot t}} \leq
	\sqrt{1 - \left( 1 - \varepsilon \right)^{2\cdot t}}\enspace .
	$$
\end{proof}

Our main lemma now follows.

\begin{lemma}[Angular Indistinguishability] \label{lemma:angular}
	Let $n \in \bbN$, $\varepsilon \in [0, 1]$, let $D$ be a distribution over (not necessarily normalized) vectors in $V \subseteq \bbC^{2^n}$, let $\varphi: V \rightarrow \bbC^{2^n}$ be a function and let $\hat{\varphi} : V \rightarrow \bbC^{2^n}$ be the normalized version of $\varphi$, $\hat{\varphi}(v) := \frac{\varphi(v)}{\norm{\varphi(v)}}$.
	Assume that for every $v \in V$, the normalization of $v$ and its $\hatphi$-image are close on the unit sphere, that is,
	$$
	\abs{\braket{\hat{v}}{\hat{\varphi}(v)}} \geq 1 - \varepsilon,
	$$
	then for all $t \in \bbN$, 
	\begin{align}\label{eq:angular}
	\td\Big(
	\bbE_{v \gets D} \big[ \left( \ket{\hat{v}}\bra{\hat{v}} \right)^{\otimes t}\big] \, , \,
	\bbE_{v \gets D} \big[ \left( \ket{\hat{\varphi}(v)}\bra{\hat{\varphi}(v)} \right)^{\otimes t} \big]
	\Big) \leq 
	\sqrt{2 t \varepsilon} \enspace .
	\end{align}
\end{lemma}

%
%
%

\begin{proof}
	We will show that for every projective measurement $\A$ (with output in $\{ -1, 1 \}$), it holds that
	\begin{equation}\label{eq:expdiff}
	\underbrace{\abs{
		\bbE_{\substack{v \gets D, \\ \text{Measurement}}}
		\Big[ 
		\A\big( \left( \ket{\hatv}\bra{\hatv} \right)^{\otimes t} \big)
		\Big] - 
		\bbE_{\substack{v \gets D, \\ \text{Measurement}}}
		\Big[ 
		\A\big( \left( \ket{\hatphi(v)}\bra{\hatphi(v)} \right)^{\otimes t} \big)
		\Big]
	}}_{\text{Denote this $\Delta$}} \leq
	2\cdot \sqrt{1 - \left( 1 - \varepsilon \right)^{2\cdot t}} \enspace ,
	\end{equation}
	where the expectation is over sampling a vector from the distribution $D$, and over the measurement outcome.
	Note that the trace distance in the lemma statement (Eq.~\eqref{eq:angular}) is equal to $\Delta/2$ (for the measurement $\A$ that maximizes this expression, the factor of $2$ is because we are now considering $\pmset$ distinguishers instead of $\binset$). Therefore, showing that Eq.~\eqref{eq:expdiff} holds will show that the trace distance in Eq.~\eqref{eq:angular} is bounded by $\sqrt{1 - \left( 1 - \varepsilon \right)^{2\cdot t}}$.
	By Bernoulli's inequality we have $\left( 1 - \varepsilon \right)^{2\cdot t} \geq 1 - 2t\varepsilon$ and therefore proving that Eq.~\eqref{eq:expdiff} holds will conclude the proof of this lemma.
	
	We use the following notation:
	\begin{itemize}
		\item 
		$\Img(\varphi)_D$ is the distribution over the image of $\varphi$, such that the probability for the element $w \in \Img(\varphi)$ is the probability to get $w = \varphi(v)$ when choosing $v \gets D$.
		
		\item 
		For $w \in \Img(\varphi)$, $\PreImg(w)_D$ is the conditional distribution of $D$, conditioned that $\varphi(v) = w$.
	\end{itemize}
	Observe that sampling $v\gets D$ is exactly like sampling $w \gets \Img(\varphi)_D$ and then the output of the process is a sample $v \gets \PreImg(w)_D$. It follows that
	\begin{align*}
    \Delta
	&
	= \bigg|
		\bbE_{w \gets \Img(\varphi)_D}\left[
		\bbE_{\substack{v \gets \PreImg(w)_D, \\ \text{Measurement}}}
		 \Big[ 
		\A\big( \left( \ket{\hatv}\bra{\hatv} \right)^{\otimes t} \big)
		\Big]
		\right] 
		\\
		& 
		\qquad - \bbE_{w \gets \Img(\varphi)_D}\left[
		\bbE_{\substack{v \gets \PreImg(w)_D, \\ \text{Measurement}}}
		\Big[ 
		\A\big( \left( \ket{\hatphi(v)}\bra{\hatphi(v)} \right)^{\otimes t} \big)
		\Big]
		\right]
	\bigg| \\
	&= \bigg|
		\bbE_{w \gets \Img(\varphi)_D}\bigg[
		\bbE_{\substack{v \gets \PreImg(w)_D, \\ \text{Measurement}}}
		\Big[ 
		\A\big( \left( \ket{\hatv}\bra{\hatv} \right)^{\otimes t} \big)
		\Big] \\
		& \qquad - \bbE_{\substack{v \gets \PreImg(w)_D, \\ \text{Measurement}}}
		\Big[ 
		\A\big( \left( \ket{\hatphi(v)}\bra{\hatphi(v)} \right)^{\otimes t} \big)
		\Big]
		\bigg]
	\bigg| \enspace.
	\end{align*}

	It is clear that the above expectation of a difference is bounded by the difference for the element $w' \in \Img(\varphi)$ that maximizes it:
	$$
	\leq \abs{
		\bbE_{\substack{v \gets \PreImg(w')_D, \\ \text{Measurement}}}
		\Big[ 
		\A\big( \left( \ket{\hatv}\bra{\hatv} \right)^{\otimes t} \big)
		\Big] - 
		\bbE_{\substack{v \gets \PreImg(w')_D, \\ \text{Measurement}}}
		\Big[ 
		\A\big( \left( \ket{\hatphi(v)}\bra{\hatphi(v)} \right)^{\otimes t} \big)
		\Big]
	} \enspace ,
	$$
	which in turn is bounded by twice the trace distance,
	$$
	\leq 2\cdot \td\Big( 
	\bbE_{v \gets \PreImg(w')_D}\big[\left( \ket{\hatv}\bra{\hatv} \right)^{\otimes t} \big],
	\bbE_{v \gets \PreImg(w')_D}\big[\left( \ket{\hatphi(v)}\bra{\hatphi(v)} \right)^{\otimes t}\big]
	\Big) \enspace .
	$$
	
	Finally, consider the following two distributions:
	\begin{itemize}
		\item 
		$D_1$: Sample $v \gets \PreImg(w')_D$ and output the normalization $\ket{\hatv}$.
		
		\item 
		$D_2$: The constant, zero-entropy distribution that always outputs $\ket{\hat{w'}}$.
	\end{itemize}
	Note that $D_1$, $D_2$ satisfy the condition that for every $\ket{\hatv}$ in the support of $D_1$ (which is the set of normalizations of $\PreImg(w')$), and every $\ket{\hatu}$ in the support of $D_2$ (which is simply $\{ \ket{\hat{w'}} \}$), we have:
	$$
	\abs{\braket{\hatv}{\hatu}} = \abs{\braket{\hatv}{\hat{w'}}} =  \abs{\braket{\hatv}{\hatphi(v)}} \geq 1 - \varepsilon \enspace .  
	$$
	This is a strong condition on the distributions $D_1$, $D_2$, which basically says that the diameter of their unified supports is small.
	This condition is formalized in Lemma \ref{claim:trace_diameter}, which in our setting implies that the trace distance between $t$ copies of $D_1$ and $t$ copies of $D_2$ is bounded by $\sqrt{1 - \left( 1 - \varepsilon \right)^{2\cdot t}}$, that is,
	$$
	2\cdot \td\Big( 
	\bbE_{v \gets \PreImg(w')_D}\big[\left( \ket{\hatv}\bra{\hatv} \right)^{\otimes t} \big],
	\bbE_{v \gets \PreImg(w')_D}\big[\left( \ket{\hatphi(v)}\bra{\hatphi(v)} \right)^{\otimes t}\big]
	\Big) \leq
	2\cdot \sqrt{1 - \left( 1 - \varepsilon \right)^{2\cdot t}} \enspace .
	$$
\end{proof}

\section{Scalable Asymptotically Random State (ARS) Generator}
\label{sec:construction}

In this section we describe a procedure that given quantum oracle access to a random classical function, efficiently samples random quantum states that are arbitrarily random (i.e. we can scale up the randomness of our sampled state and make it increasingly harder to distinguish from a random quantum state, for an increasing number of output copies, without increasing the number of qubits in the state) and can generate multiple copies of a state when executed multiple times with oracle access to the same function.
More formally, we describe a sampling procedure with the following inputs:

\begin{itemize}
	\item 
	$1^n$: Number of wanted qubits in the output state.
	
	\item 
	$1^\secp$: Security parameter that measures "how random" the output state is going to be (i.e. how hard will it be to distinguish $t$ copies of the sampled state from $t$ copies of a random quantum state, as a function of $\secp, t$).
	
	\item 
	Quantum oracle access to $U_f$: For a function $f: \{ 0, 1 \}^{n} \rightarrow \{ 0, 1 \}^{\poly(n, \secp)}$ (for some polynomial $\poly(\cdot)$, specified later), the sampling procedure gets oracle access to the unitary mapping $U_f$ of $f$.
	
\end{itemize}

The formal statement that explains how to construct a scalable ARS generator follows.

\begin{theorem} [Scalable ARS Generator Construction] \label{thm:main}
	There exists a scalable ARS generator $\Gen$ that for every $n \in \bbN$ number of qubits, $5 \leq \secp \in \bbN$ security parameter and $t \in \bbN$ number of copies, satisfies the following trace distance bound,
	$$
	\td\Big(
	D_1  , 
	D_2
	\Big) \leq \left( t + 8 \right)\cdot e^{-\secp} + \left( 5\sqrt{t} + \secp + 1 \right)\cdot 2^{-\secp} + 2\cdot \left( \frac{8}{10} \right)^{\secp} \enspace ,
	$$
	where the distributions $D_1, D_2$ are defined as follows:
	\begin{itemize}
		\item $D_1 : $
		Sample $\tilde{f} \gets \left( \{ 0, 1 \}^{\poly(n, \secp)} \right)^{\{ 0, 1 \}^{n}}$, execute $t$ times the generation algorithm $\Gen^{U_{\tilde{f}}}(1^{n}, 1^{\secp})$ and output the $t$ output states.
		
		\item $D_2 : $
		Sample $\ket{\psi}$ a random $n$-qubit state and output $\ket{\psi}^{\otimes t}$.
	\end{itemize}
\end{theorem}

\begin{proof}
	We start with describing the procedure of $\Gen^{U_{\tilde{f}}}(1^n, 1^\secp)$.
	First, we denote $\varepsilon := 2^{-n-\secp}$, $B := \ceil{2\sqrt{n + \secp}}$ and set the polynomial $\poly(n, \secp)$ that denotes the output size of $\tilde{f}$ to be $\secp \cdot r(\varepsilon, B)$, where $r(\varepsilon, B)$ is the randomness complexity of the rounded Gaussian sampler $\gauss{\varepsilon, B}$.
	Given the oracle access to $\tilde{f} \in \left( \{ 0, 1 \}^{\secp \cdot r(\varepsilon, B)} \right)^{\{ 0, 1 \}^{n}}$, the algorithm starts with deciding on a different function $f \in \left( \{ 0, 1 \}^{r(\varepsilon, B)} \right)^{\{ 0, 1 \}^{n}}$ that it is going to use.
	
	In what follows, denote $N := 2^n$, for a function $h \in \left( \{ 0, 1 \}^{r(\varepsilon, B)} \right)^{\{ 0, 1 \}^{n}}$ denote by $v^h$ the vector that is created by rounded Gaussian sampling with $h$, that is, $\forall x \in \{ 0, 1 \}^n, v^h_x := \gauss{\varepsilon, B}(h(x))$.
	We think of $\tilde{f}$, that has an output length of $\secp \cdot r(\varepsilon, B)$, as $\secp$ different functions, each having an output length of $r(\varepsilon, B)$.
	Specifically, for $i \in [\secp]$ define the function $f_i \in \left( \{ 0, 1 \}^{r(\varepsilon, B)} \right)^{\{ 0, 1 \}^{n}}$ as the function that for input $x \in \{ 0, 1 \}^n$ outputs the $i$-th packet of $r(\varepsilon, B)$ bits from $\tilde{f}(x)$.
	
	\noindent The procedure of $\Gen$ follows.
	\begin{enumerate}
		\item 
		Decide on a function $f \in \left( \{ 0, 1 \}^{r(\varepsilon, B)} \right)^{\{ 0, 1 \}^{n}}$:
		\begin{itemize}
			\item 
			If $N > \secp$, we actually use only the first $r(\varepsilon, B)$ bits of the output of $\tilde{f}$.
			That is, $f$ is simply $f_1$.
			
			\item
			If $N \leq \secp$, iterate for $i \in [\secp]$:
			\begin{itemize}
				\item 
				Compute the vector $v^{f_{i}}$ by applying $\gauss{\varepsilon, B}$ to each of the $N$ outputs of $f_i$.
				If $\norm{v^{f_i}} \geq \frac{\sqrt{N}}{2}$, denote $f := f_i$ and halt the loop.\footnote{Note that all steps here can be done efficiently in $\secp$, because $\secp \geq N = 2^n$.} 
			\end{itemize}
			
			If you executed all iterations and did not get a function $f_i$ s.t. $\norm{v^{f_i}} \geq \frac{\sqrt{N}}{2}$, halt and output $\ket{0^{n}}$ (as a sign of failure).
		\end{itemize}
		
		\item
		Given $f$ execute $\bvs^{U_{v^f}}(\sqrt{2}\cdot B, 8\cdot \secp\cdot B^2)$ and output the $n$-qubit quantum state generated by $\bvs$.
	\end{enumerate}
	
	
	We now need to show that the distributions $D_1$ and $D_2$ are close in trace, this will be done by a hybrid argument, that is, we will consider hybrid distributions, starting from the distribution generated by our construction, and approaching the distribution over truly random quantum states (that is, $t$ copies of a random state).
	We will explain why each pair of consecutive distributions are close and at the end use the triangle inequality of trace distance to bound everything together.
	
	In what follows we will use the fact that process $D_1$ is exactly the following:
	\begin{enumerate}
	    \item A first step where we sample $v^f$. 
	    \begin{itemize}
	        \item
	        In the case $N > \secp$, we sample $f \gets \left( \{ 0, 1 \}^{r(\varepsilon, B)} \right)^{\{ 0, 1 \}^{n}}$ once and $v^f$ is determined.
	        
	        \item
	        In the case $N \leq \secp$, we execute $t$ tries. In each try, we sample $f \gets \left( \{ 0, 1 \}^{r(\varepsilon, B)} \right)^{\{ 0, 1 \}^{n}}$, compute the vector $v^f$ and check whether $\norm{v^{f}} \geq \frac{\sqrt{N}}{2}$.
	        If none of the trials succeed, the output is $\ket{0^n}$.
	    \end{itemize}
	    
	    \item A second step where we (try to) generate $t$ copies of the quantum state $\hatv^f$, by executing $\bvs^{U_{v^f}}(\sqrt{2}\cdot B, 8\cdot \secp\cdot B^2)$, $t$ times.
	\end{enumerate}
	Consider the following sampling processes.
	
	\begin{itemize}
		\item $P_1: $ We start with the process $D_1$.
		
		\item $P_2: $ 
		We would like to think about a more natural distribution for the vector $v^f$, and move to the output distribution of the Gaussian sampler algorithm.
		In this process every sampling of $v^f$ is swapped with a sampling $v \gets \left( \gauss{\varepsilon, B} \right)^{N}$, the (multidimensional) output distribution of the Gaussian sampler.
		This means in particular that for the case $N \leq \secp$ we repeatedly sample $v \gets \left( \gauss{\varepsilon, B} \right)^{N}$ and check that $\norm{v} \geq \frac{\sqrt{N}}{2}$ (instead of sampling $f$ and check that $\norm{v^f} \geq \frac{\sqrt{N}}{2}$), and if we didn't get $\norm{v} \geq \frac{\sqrt{N}}{2}$ in any of our tries then we also output $\ket{0^{n}}$.
		After obtaining $v$, we carry on regularly (that is, we perform $t$ executions of $\bvs^{U_{v}}(\sqrt{2}\cdot B, 8\cdot \secp\cdot B^2)$).
		This process yields exactly the same distribution as $P_1$.
		
		\item $P_3: $ In this step we will move to the actual rounded Gaussian distribution, rather than the output distribution of the Gaussian sampler (which doesn't necessarily sample exactly from the rounded Gaussian distribution).
		In this process everything is identical to the last process, with the exception that every sampling $v \gets \left( \gauss{\varepsilon, B} \right)^{N}$ is swapped to sampling $v \gets \distribution{N}{\varepsilon, B}$. 
		To move to this distribution we will use the correctness guarantee of the sampler algorithm $\gauss{\varepsilon, B}$.
		
		\item $P_4: $ In order to use the correctness guarantee of the algorithm $\bvs$ and argue that it generates copies of the quantum state $\ket{\hatv}$ correctly, we will need our sampled vector $v$ to have large norm.
		In this point we would like to sample a vector that's always long.
		Therefore, instead of performing the first step where we sample $v \gets \distribution{N}{\varepsilon, B}$ (once, in case $N > \secp$, or $\secp$ times until $v$ is long enough, in the case $N \leq \secp$), we sample {\it once}, $v \gets \widetilde{\distribution{N}{\varepsilon, B}}$ from the {\it conditional} distribution of $\distribution{N}{\varepsilon, B}$, with the condition that $v$ is long enough: $\norm{v} \geq \frac{\sqrt{N}}{2}$.
		After the sampling of $v$ we carry on regularly to execute $\bvs$ $t$ executions.
		
		\item $P_5: $ We now move to a process where after the vector is sampled, it just outputs $t$ copies of its corresponding quantum state.
		This process is identical to the last one, with the exception that after sampling $v \gets \widetilde{\distribution{N}{\varepsilon, B}}$, instead of executing $\bvs$ $t$ times, the process simply outputs the quantum state $\ket{\hatv}^{\otimes t}$.
		To move to the process we'll use the success probability of the algorithm $\bvs$.
		
		\item $P_6: $ We currently manage to sample vectors and generate a $t$-tensor of their quantum state, and we would like to start smoothing the distribution that our vector $v$ is coming from, for it to get closer to a random unit vector.
		In particular, the distribution over random quantum states is continuous, and our processes so far yields only a discrete distribution.
		This step is a preparation for the next ones, and is intended so that we will consider the same output distribution as before, but generated by a different process: The sampling $v \gets \widetilde{\distribution{N}{\varepsilon, B}}$ is swapped to sampling $v'$ from the conditional distribution of the standard (and continuous) Gaussian distribution $v' \gets \normalCompDist^{N}$ with the same condition that the rounding $R_{(\varepsilon,  B)}(v')$ of $v'$ satisfies the norm condition $\norm{R_{(\varepsilon,  B)}(v')} \geq \frac{\sqrt{N}}{2}$.
		We then continue as before and output a $t$-tensor of the quantum state of the rounded vector, that is, $\ket{\widehat{R_{(\varepsilon,  B)}}(v')}^{\otimes t}$.
		Note that this process samples from a continuous distribution but still \emph{outputs} a discrete distribution.
		These two processes yield identical distributions due to the definition of the rounded Gaussian distribution.
		
		\item $P_7: $ In order move to a continuous distribution we will use Lemma \ref{lemma:angular}, and for this we will need that $v'$ will not only yield a long rounding $R_{(\varepsilon,  B)}(v')$, but we'll need $v'$ to be physically close to $R_{(\varepsilon,  B)}(v')$.
		In this process, instead of sampling from $\normalCompDist^{N}$ with the only condition $\norm{R_{(\varepsilon,  B)}(v')} \geq \frac{\sqrt{N}}{2}$ we add a second condition: $\forall i \in [N] : |v'_i| \leq B$ (recall that for complex numbers with $B$-bounded absolute value, their rounding by $R_{(\varepsilon, B)}(\cdot)$ will be $\sqrt{2}\cdot\varepsilon$-close).
		
		\item $P_8: $ We finally move to a continuous output distribution.
		$v'$ is sampled as before (from the conditional distribution of the continuous Gaussian distribution, with the conditions that $v'$ is bounded per-coordinate and that the rounding of $v'$ is long enough), but this time we output the quantum state without rounding it, that is, $\ket{\hatv'}^{\otimes t}$ (rather then $\ket{\widehat{R_{(\varepsilon,  B)}}(v')}^{\otimes t}$). 
		For the pass between these two distributions we use  Lemma \ref{lemma:angular}.
		
		\item $P_9: $ Our last process yielded a continuous distribution over quantum states, but the uniform distribution over quantum states is not only continuous, it is also {\bf spherically symmetric} (a property that our current distribution does not possess).
		To move to a spherically symmetric distribution it's tempting to simply cancel our conditions on $v'$; this will indeed yield a spherically symmetric distribution (even the most obvious one - we know that a normalized random Gaussian vector is exactly a random unit vector) but it is not immediately clear how to show that such distribution is close to ours\footnote{Note that the probability mass of the vectors that do not satisfy the two conditions is negligible, but in $N$. Our main challenge in this construction is to show negligible statistical distance in $\secp$.}.
		Our trick will be to observe that the conditional distribution of the Gaussian distribution where the condition is a \emph{uniform} lower bound (as a side note, can also be upper bound) on the sampled vectors is also a spherically symmetric distribution, and we'll move to this spherically symmetric distribution in two steps.
		
		In this step we'll tweak our condition on the length of $v'$: we will ask that $\norm{v'} \geq \frac{\sqrt{N}}{2}$ rather than the condition that the {\it rounded vector} is long enough $\norm{R_{(\varepsilon, B)}(v')} \geq \frac{\sqrt{N}}{2}$ (the other condition about the per-cordinate bound stays the same).
		
		\item $P_{10}: $ 
		In this step we'll drop the per-coordinate-bound condition from the sampling of $v'$, that is, $v'$ is sampled from the conditional distribution of $\normalCompDist^{N}$ with the only condition being $\norm{v'} \geq \frac{\sqrt{N}}{2}$.
		The output is then again the quantum state $\ket{\hatv'}^{\otimes t}$.
		
		\item $P_{11}: $ Our last step will be the distribution $D_2$.
		By the spherical symmetry of the last distribution, it is actually identical to this one.
	\end{itemize}
	
	We now explain why each pair of consecutive distributions are close in trace, and recall that we need to show this closeness in trace as a {\it negligible function of $\secp$}, rather than of $n$ (or even $N$, as we always keep in mind the case where $n = 1, N = 2$).
	A useful notion we will repeatedly use is, for distributions over quantum states, we can always consider their "vector distributions", which are simply distributions over the vectors that describe the output quantum states (with the same corresponding probabilities).
	By giving an upper bound on the statistical distance between a pair of such vector distributions, by Lemma \ref{lemma:sd_td} we also obtain the same upper bound on the trace distance between the two original quantum-state distributions, regardless of the number of output copies of the sampled quantum state.
	
	\begin{itemize}
		\item $\td\left( P_1, P_2 \right) = 0 :$
		Sampling a vector $v^f$ by sampling a random function $f$ and then applying the sampler $\gauss{\varepsilon, B}$ to each of its outputs is by definition exactly like sampling a vector $v \gets \left( \gauss{\varepsilon, B} \right)^{N}$, thus the vector distributions of the quantum-state distributions $P_1, P_2$ have statistical distance 0 and thus in particular the above trace distance follows.
		
		\item  $\td\left( P_2, P_3 \right) \leq \secp \cdot 2^{-\secp}:$
		Consider the vector distributions of $P_2, P_3$.
		In the process $P_2$, in any of the cases (either $N > \secp$ or $N \leq \secp$) we sample at most $N \cdot \secp$ times from the sampler $\gauss{\varepsilon, B}$, and in $P_3$ we perform exactly the same sampling pattern but from the distribution $\SingleDistribution{\varepsilon, B}$.
		By the correctness of the sampling algorithm $\gauss{\varepsilon, B}$, the statistical distance between the output of $\gauss{\varepsilon, B}$ and $\SingleDistribution{\varepsilon, B}$ is bounded by $\varepsilon:=2^{-n-\secp}$, and thus the statistical distance of sampling $N\cdot \secp$ times from these two is bounded by $N\cdot \secp \cdot \varepsilon = \secp \cdot 2^{-\secp}$, which is a bound on the statistical distance between the vector distributions of $P_2, P_3$, and thus we get the above bound.
		
		\item $\td\left( P_3, P_4 \right) \leq 2\cdot\left( \frac{78}{100} \right)^{\secp}:$
		Consider the vector distributions of $P_3, P_4$, and denote them by $\tilde{P_3}, \tilde{P_4}$, accordingly.
		Observe that $\tilde{P_4}$ is a conditional distribution of $\tilde{P_3}$: In the case $N > \secp$, it is conditioned on that the sampled $v$ satisfies the norm condition, and in the case $N \leq \secp$ is it conditioned that in one of the $\secp$ samplings, the sampled vector $v$ satisfies the norm condition.
		This implies that the statistical distance between the two vector distributions is bounded by the probability that the condition does not hold (Fact \ref{fact:statistical_condition}).
		
		By Lemma \ref{lemma:rounded_gaussian_vectors_long}, the probability that a rounded Gaussian vector with $N$ coordinates and rounding parameters $\varepsilon := 2^{-n-\secp}, B := \ceil{2\sqrt{n + \secp}}$ does not satisfy the balance condition, is bounded by $e^{-\frac{N}{4}} + e^{-\secp}$.
		In the case where $N > \secp$ this probability is bounded by,
		$$
		e^{-\frac{\secp}{4}} + e^{-\secp} \leq
		2\cdot e^{-\frac{\secp}{4}} <
		2\cdot \left( \frac{78}{100} \right)^{\secp} \enspace .
		$$
		Also, in the case where $N \leq \secp$, we perform $\secp$ independent samplings of $v$ and thus the probability that we fail (in all $\secp$ times) is bounded by,
		$$
		\left( e^{-\frac{N}{4}} + e^{-\secp} \right)^\secp \underset{\big( N = 2^n \geq 2 \; , \; \secp \geq 3 \big)}{\leq}
		\left( e^{-\frac{2}{4}} + e^{-3} \right)^\secp <
		\left( \frac{7}{10} \right)^\secp \enspace .
		$$
		In any of the cases, the probability is bounded by $2\cdot\left( \frac{78}{100} \right)^{\secp}$.
		
		\item $\td\left( P_4, P_5 \right) \leq t\cdot e^{-\secp}:$
		Let $\tilde{P_4}, \tilde{P_5}$ be the vector distributions of $P_4, P_5$.
		Observe that $\tilde{P_5}$ is a conditional distribution of $\tilde{P_4}$, conditioned on the algorithm $\bvs$ succeeding $t$ times to generate the state $\ket{\hatv}$, which implies that the statistical distance between them is bounded by the probability that in one of $t$ executions, $\bvs$ failed at least once (Fact \ref{fact:statistical_condition}).
		By a union bound this probability is in turn bounded by $t$ times the probability that $\bvs^{U_v}(\sqrt{2}\cdot B, 8\cdot \secp\cdot B^2)$ failed (once).
		In both distributions, $v$ is sampled from the conditional distribution $\widetilde{\distribution{N}{\varepsilon, B}}$ and thus satisfies $\norm{v} \geq \frac{\sqrt{N}}{2}$ always (also, by the definition of the $(\varepsilon, B)$-rounded distribution, $\forall i \in [N] : |v_i| \leq \sqrt{2} \cdot B$).
		By Claim \ref{claim:amp_bvs_success}, the probability that $\bvs^{U_v}(\sqrt{2}\cdot B, 8\cdot \secp\cdot B^2)$ fails to generate $\ket{\hatv}$ is bounded by $e^{-\secp}$, and it follows that the statistical distance between $\tilde{P_4}, \tilde{P_5}$ is bounded by $t\cdot e^{-\secp}$ and thus the same bound applies for the original distributions $P_4, P_5$.
		
		\item $\td\left( P_5, P_6 \right) = 0:$
		By definition of the rounded Gaussian distribution, sampling $w$ from $\distribution{N}{\varepsilon, B}$ is like sampling $w' \gets \normalCompDist^{N}$ and outputting $R_{(\varepsilon, B)}(w')$.
		Note that it follows that sampling $v \gets \widetilde{\distribution{N}{\varepsilon, B}}$ (and outputting $v$) is identical to sampling $v' \gets \normalCompDist^{N}$ conditioned on $\norm{R_{(\varepsilon,  B)}(v')} \geq \frac{\sqrt{N}}{2}$ (and outputting $R_{(\varepsilon,  B)}(v')$).
		
		\item $\td\left( P_6, P_7 \right) \leq 4\cdot e^{-\secp}:$
		The vector distribution of $P_7$ is by its definition a conditional distribution of (the vector distribution of) $P_6$, with the condition that the sampled $v'$ satisfies the per-coordinate bound condition: $\forall i \in [N] : |v'_i| \leq B$.
		By Corollary \ref{cor:conditional_vectors_coordinates_bounded}, the probability that the condition does not hold is bounded by $4\cdot e^{-\secp}$, and thus the bound on the statistical distance between the vector distributions (and also the quantum-state distributions) follows.
		
		\item $\td\left( P_7, P_8 \right) \leq 6\sqrt{2}\cdot \sqrt{t} \cdot N^{-1}\cdot 2^{-\secp}:$
		We now use lemma \ref{lemma:angular}, but we will first understand why we can use it.
		Denote by $D$ the distribution that $v'$ is sampled from and by $V$ the support of $D$.
		Let $v' \in V$, and by the definition of $V$ we have $\forall i \in [N] : |v'_i| \leq B$, and thus in every entry of $v'$ both real and imaginary parts are $B$-bounded, and by the definition of the rounding function both are rounded upwards when $R_{(\varepsilon, B)}(\cdot )$ is applied, and furthermore are rounded upwards by at most $\varepsilon$.
		This in turn implies that $\forall i \in [N] : |v'_i - R_{(\varepsilon, B)}(v')_i| \leq \sqrt{2} \cdot\varepsilon$ (note that we have a factor of $\sqrt{2}$ because we are dealing with complex numbers).
		Also by the definition of $V$, we have $\norm{R_{(\varepsilon, B)}(v')} \geq \frac{\sqrt{N}}{2}$.
		The above implies that by Lemma \ref{lemma:long_vectors}, $\abs{\braket{\widehat{R_{(\varepsilon,  B)}}(v')}{\hatv'}} \geq 1 - 36\cdot \varepsilon^2$.
		
		Finally, denote by $\varphi(\cdot) : V \rightarrow \bbC^{N}$ the rounding function $R_{(\varepsilon, B)}(\cdot)$, and observe that the distribution $P_7$ is exactly sampling from $D$, applying $\varphi$, normalizing and outputting $t$ copies of the corresponding quantum state, and the distribution $P_8$ is the same, except that we don't apply $\varphi$.
		Because of the above and because we have $\forall v' \in V : \abs{\braket{\widehat{R_{(\varepsilon,  B)}}(v')}{\hatv'}} \geq 1 - 36\cdot \varepsilon^2$, by Lemma \ref{lemma:angular}
		the trace distance between $P_7$, $P_8$ is bounded by $\sqrt{2 \cdot t \cdot 36\cdot \varepsilon^2} = 6\sqrt{2}\cdot \sqrt{t} \cdot 2^{-n-\secp}$.
		
		\item $\td\left( P_8, P_9 \right) \leq 4\cdot 2^{-\secp} \cdot 6^{-N}:$
		Observe that the vector distribution of $P_{9}$ is a conditional distribution of the vector distribution of $P_{8}$, with the condition $\norm{v'} \geq \frac{\sqrt{N}}{2}$.
		As usual, the statistical distance between the distributions is bounded by the probability that the condition does not hold, and by Corollary \ref{cor:balanced_vectors_not_in_layer} it follows that this probability is bounded by $4\cdot 2^{-\secp} \cdot 6^{-N}$ and thus so is the statistical distance.
		
		\item $\td\left( P_9, P_{10} \right) \leq 4\cdot e^{-\secp}:$
		This trace closeness is by very similar reasoning as in the explanation for the upper bound on $\td\left( P_6, P_7 \right)$.
		Specifically, $P_9$ is a conditional distribution of $P_{10}$ with the condition that the vector $v'$ is per-coordinate-bounded.
		The probability that this condition does not hold is bounded by $4\cdot e^{-\secp}$ according to Corollary \ref{cor:2nd_conditional_vectors_coordinates_bounded}.
		
		\item$\td\left( P_{10}, P_{11} \right) = 0:$
		Recall that $P_{11}$ is defined by sampling a random unit vector $v$ from $\bbC^{N}$ and outputting a $t$-tensor of its corresponding quantum state $\ket{\hatv}$.
		Also, it is a known property of the Gaussian distribution that a normalized random Gaussian vector distributes exactly like a random unit vector, this follows from a known property of the Gaussian distribution, which is that it is spherically symmetric i.e. for any unitary transformation $U$, sampling a Gaussian vector and applying $U$ distributes like sampling a Gaussian vector without applying $U$.
		The spherical symmetry property in fact also implies that any conditional distribution of the Gaussian distribution, with the condition that the sampled vector has norm at least $k$ (for any $k \in \bbR$), is also spherically symmetric.
		It follows that a normalization of such conditional distribution of the Gaussian distribution also distributes like a random unit vector.
		
		The distribution $P_{10}$ is captured by the description of this conditional distribution, with $k$ being $\frac{\sqrt{N}}{2}$.
		It follows that the vector distributions of $P_{10}$ and $P_{11}$ are identical and have statistical distance of 0, and thus by Lemma \ref{lemma:sd_td} the trace distance between $P_{10}$ and $P_{11}$ follows.
		\qedhere
	\end{itemize}
\end{proof}

\subsection{Statistical Properties of Multidimensional Gaussian Distribution} \label{subsec:gaussian}

In this subsection we state some useful properties of the Gaussian distribution, relevant to our construction and that support the proof of the main Theorem \ref{thm:main}. We start with quoting \cite[Corollary 2.3]{barvinok2005math}, and follow with a sequence of lemmas and corollaries that are technically involved but conceptually straightforward.

\begin{lemma}[Gaussian Vectors are Almost Always Long {\cite[Corollary 2.3]{barvinok2005math}}] \label{lemma:gaussian_vectors_long}
	For $m \in \bbN$, $\varepsilon > 0$ we have,
	$$
	\Pr_{u \gets \normalDist^m}\left[ \norm{u}^2 \geq \left( 1 - \varepsilon \right)\cdot m \right] \geq 1 - e^{-\frac{\varepsilon^2 \cdot m}{4}} \enspace .
	$$
\end{lemma}

We'll see that the probability for a sampled vector to be bounded per-coordinate is overwhelming, and thus it can be shown to be balanced with high probability.

\begin{lemma}[Gaussian Vectors are Almost Always Bounded Per-Coordinate] \label{lemma:gaussian_vectors_coordinate_bounded}
	Let $n \in \bbN$, $\secp \in \bbN$ and $B \geq 2\sqrt{n + \secp}$.
	Then,
	$$
	\Pr_{v \gets \normalCompDist^{2^n}}\Big[ \exists i \in [2^n] : |v_i| > B \Big] < e^{-\secp} \enspace .
	$$
\end{lemma}

\begin{proof}
	The proof follows by a calculation and by known properties of the Gaussian distribution.
	
    \begin{align}
	\Pr_{v \gets \normalCompDist^{2^n}}\Big[ \exists i \in [2^n] : |v_i| > B \Big] &\underset{(\text{union bound})}{\leq}
	2^n \cdot \Pr_{z \gets \normalCompDist}\Big[ |z| > B \Big] \nonumber\\
	&= 2^n \cdot \Pr_{y \gets \normalDist^2}\Big[ \sqrt{y^2_1 + y^2_2} > B \Big] \nonumber\\
	&\underset{(\text{union bound})}{\leq} 2^n \cdot 2 \cdot \Pr_{x \gets \normalDist}\bigg[ |x| > \frac{B}{\sqrt{2}} \bigg] \nonumber\\
	&= 2^n \cdot 4\cdot \Pr_{x \gets \normalDist}\bigg[ x > \frac{B}{\sqrt{2}} \bigg] \nonumber\\
	&< 2^{n} \cdot 4\cdot \frac{1}{B\sqrt{\pi}} \cdot e^{-\frac{B^2}{2}\cdot \frac{1}{2}} \label{eq:gaussian}\\
	&= \frac{4}{B \sqrt{\pi}} \cdot 2^{n} \cdot e^{-\frac{B^2}{4}} \nonumber\\
	&\underset{\left( B \geq 2\cdot \sqrt{n + \secp} \right)}{<} e^{- \secp} \nonumber \enspace ,
	\end{align}

	where \ref{eq:gaussian} follows from a well known upper bound on the q-function of the Gaussian distribution:
	$$
	\forall t \geq 0 : \Pr_{x \gets \normalDist}\big[ x > t \big] < \frac{1}{t\sqrt{2\cdot \pi}} \cdot e^{-\frac{t^2}{2}} \enspace .
	$$
\end{proof}

The corollary about Gaussian vectors being balanced follows.

\begin{corollary}[Gaussian Vectors are Almost Always Balanced] \label{lemma:gaussian_vectors_balanced}
	Let $n \in \bbN$, $\secp \in \bbN$ and $B \geq 2\sqrt{n + \secp}$.
	Then we have,
	$$
	\Pr_{v \gets \normalCompDist^{2^n}}\left[ \left( \norm{v} \geq \frac{\sqrt{2^n}}{2} \right) \land \left( \forall i \in [2^n] : |v_i| \leq B \right) \right] > 1 - \left( e^{-\frac{2^n}{4}} + e^{-\secp} \right) \enspace .
	$$
\end{corollary}

\begin{proof}
	\begin{align*}
	1 - e^{-\frac{2^n}{4}} &\leq 
	1 - e^{-\left( \frac{7}{8} \right)^2 \cdot \frac{2\cdot 2^n}{4}} \\
	&\underset{(\text{Lemma \ref{lemma:gaussian_vectors_long}})}{\leq}
	\Pr_{u \gets \normalDist^{2\cdot 2^n}}\left[ \norm{u}^2 \geq \left( 1 - \frac{7}{8} \right)\cdot \left( 2\cdot 2^n \right) \right] \\
	&= \Pr_{u \gets \normalDist^{2\cdot 2^n}}\left[ \norm{u} \geq \frac{\sqrt{2^n}}{2} \right] \\
	&= \Pr_{v \gets \normalCompDist^{2^n}}\left[ \norm{v} \geq \frac{\sqrt{2^n}}{2} \right] \\
	&= \Pr_{v \gets \normalCompDist^{2^n}}\left[ \left( \norm{v} \geq \frac{\sqrt{2^n}}{2} \right) \land \left( \exists i \in [2^n] : |v_i| > B \right) \right] \\
	&+ \Pr_{v \gets \normalCompDist^{2^n}}\left[ \left( \norm{v} \geq \frac{\sqrt{2^n}}{2} \right) \land \left( \forall i \in [2^n] : |v_i| \leq B \right) \right] \\
	&\leq \Pr_{v \gets \normalCompDist^{2^n}}\Big[ \exists i \in [2^n] : |v_i| > B \Big] \\
	&+ \Pr_{v \gets \normalCompDist^{2^n}}\left[ \left( \norm{v} \geq \frac{\sqrt{2^n}}{2} \right) \land \left( \forall i \in [2^n] : |v_i| \leq B \right) \right] \\
	&\underset{(\text{Lemma }\ref{lemma:gaussian_vectors_coordinate_bounded})}{<} e^{-\secp} +
	\Pr_{v \gets \normalCompDist^{2^n}}\left[ \left( \norm{v} \geq \frac{\sqrt{2^n}}{2} \right) \land \left( \forall i \in [2^n] : |v_i| \leq B \right) \right] \enspace ,
	\end{align*}
	
	and the above inequalities imply our wanted inequality,
	$$
	1 - \left( e^{-\frac{2^n}{4}} + e^{-\secp} \right) <
	\Pr_{v \gets \normalCompDist^{2^n}}\left[ \left( \norm{v} \geq \frac{\sqrt{2^n}}{2} \right) \land \left( \forall i \in [2^n] : |v_i| \leq B \right) \right] \enspace .
	$$
\end{proof}

The last corollary implies that vectors which are sampled from rounded Gaussian distribution are also long, if the tail-cut parameter is sufficiently large.

\begin{corollary}[Rounded Gaussian Vectors are Almost Always Long] \label{lemma:rounded_gaussian_vectors_long}
	Let $n \in \bbN$, $\secp \in \bbN$, $\varepsilon > 0$ and $B \geq 2\sqrt{n + \secp}$.
	Then we have,
	$$
	\Pr_{v \gets \distribution{2^n}{\varepsilon, B}}\left[ \norm{v} \geq \frac{\sqrt{2^n}}{2} \right] > 1 - \left( e^{-\frac{2^n}{4}} + e^{-\secp} \right) \enspace .
	$$
\end{corollary}

\begin{proof}
	The calculation of the probability follows.
	\begin{align}
	\Pr_{v \gets \distribution{2^n}{\varepsilon, B}}\left[ \norm{v} \geq \frac{\sqrt{2^n}}{2} \right]
	&= \Pr_{u \gets \normalCompDist^{2^n}}\left[ \norm{R_{(\varepsilon, B)}(u)} \geq \frac{\sqrt{2^n}}{2} \right] \nonumber\\
	& \geq \Pr_{u \gets \normalCompDist^{2^n}}\left[ \left( \norm{u} \geq \frac{\sqrt{2^n}}{2} \right) \land \big( \forall i \in \left[ 2^n \right] : |u_i| \leq B \big) \right] \label{eq:rounding}\\ &\underset{(\text{Corollary }\ref{lemma:gaussian_vectors_balanced})}{>} 
	1 - \left( e^{-\frac{2^n}{4}} + e^{-\secp} \right) \enspace \nonumber,
	\end{align}
	
	where \ref{eq:rounding} is due to properties of the rounding function; if $|u_i| \leq B$ then in particular $\abs{\text{Re}(u_i)} \leq B$, $\abs{\text{Im}(u_i)} \leq B$, and then it is necessarily the case that $\forall i \in [2^n] : |u_i| \leq |R_{(\varepsilon, B)}(u_i)|$ which then implies $\norm{u} \leq \norm{R_{(\varepsilon, B)}(u)}$, and thus $\frac{\sqrt{2^n}}{2} \leq \norm{R_{(\varepsilon, B)}(u)}$ follows from $\frac{\sqrt{2^n}}{2} \leq \norm{u}$.
\end{proof}

\paragraph{Statistical Properties of Conditional Distributions.}
During the proof of Theorem \ref{thm:main} we also argue about statistical properties of conditional distributions that are related to the Gaussian distribution.
For these arguments we derive a few corollaries.

\begin{corollary} \label{cor:conditional_vectors_coordinates_bounded}
	Let $n \in \bbN$, $\secp \in \bbN$ s.t. $\secp \geq 3$, $\varepsilon > 0$, $B \geq 2\sqrt{n + \secp}$ and let $D$ be the conditional distribution of $v \gets \normalCompDist^{2^n}$ with the condition that $\norm{R_{(\varepsilon, B)}(v)} \geq \frac{\sqrt{2^n}}{2}$.
	Then,
	$$
	\Pr_{v \gets D}\Big[ \exists i \in [2^n] : |v_i| > B \Big] < 4 \cdot e^{-\secp} \enspace .
	$$
\end{corollary}

\begin{proof}
	We have,
	\begin{align*}
	e^{-\secp}
	&\underset{(\text{Lemma }\ref{lemma:gaussian_vectors_coordinate_bounded})}{>}
	\Pr_{v \gets \normalCompDist^{2^n}}\Big[ \exists i \in [2^n] : |v_i| > B \Big] \\
	&> \Pr_{v \gets \normalCompDist^{2^n}}\bigg[ \bigg( \norm{R_{(\varepsilon, B)}(v)} \geq \frac{\sqrt{2^n}}{2} \bigg) \land \big( \exists i \in [2^n] : |v_i| > B \big) \bigg] \\
	&= \Pr_{v \gets \normalCompDist^{2^n}}\bigg[ \norm{R_{(\varepsilon, B)}(v)} \geq \frac{\sqrt{2^n}}{2} \bigg] \cdot
	\Pr_{v \gets D}\Big[ \exists i \in [2^n] : |v_i| > B \Big] \\
	&= \Pr_{u \gets \distribution{2^n}{\varepsilon, B}}\bigg[ \norm{u} \geq \frac{\sqrt{2^n}}{2} \bigg] \cdot
	\Pr_{v \gets D}\Big[ \exists i \in [2^n] : |v_i| > B \Big] \\ 
	&\underset{(\text{Corollary }\ref{lemma:rounded_gaussian_vectors_long})}{>}
	\Big( 1 - \left( e^{-\frac{2^n}{4}} + e^{-\secp} \right) \Big) \cdot \Pr_{v \gets D}\Big[ \exists i \in [2^n] : |v_i| > B \Big] \\
	&\underset{(2^n \geq 2 \; , \; \secp \geq 3)}{\geq} \Big( 1 - \left( e^{-\frac{1}{2}} + e^{-3} \right) \Big) \cdot \Pr_{v \gets D}\Big[ \exists i \in [2^n] : |v_i| > B \Big] \\
	&> \left( \frac{3}{10} \right) \cdot \Pr_{v \gets D}\Big[ \exists i \in [2^n] : |v_i| > B \Big] \enspace .
	\end{align*}
	
	The above implies in particular,
	$$
	e^{-\secp} \cdot \left( \frac{10}{3} \right) > \Pr_{v \gets D}\Big[ \exists i \in [2^n] : |v_i| > B \Big] \enspace ,
	$$
	which implies our wanted inequality.
\end{proof}

\begin{corollary} \label{cor:2nd_conditional_vectors_coordinates_bounded}
	Let $n \in \bbN$, $\secp \in \bbN$, $\varepsilon > 0$, $B \geq 2\sqrt{n + \secp}$ and let $D$ be the conditional distribution of $v \gets \normalCompDist^{2^n}$ with the condition that $\norm{v} \geq \frac{\sqrt{2^n}}{2}$ (unlike the condition $\norm{R_{(\varepsilon, B)}(v)} \geq \frac{\sqrt{2^n}}{2}$ from Corollary \ref{cor:conditional_vectors_coordinates_bounded}).
	Then,
	$$
	\Pr_{v \gets D}\Big[ \exists i \in [N] : |v_i| > B \Big] < 4\cdot e^{-\secp} \enspace .
	$$
\end{corollary}

\begin{proof}
	We have,
	\begin{align*}
	e^{-\secp}
	&\underset{(\text{Lemma }\ref{lemma:gaussian_vectors_coordinate_bounded})}{>}
	\Pr_{v \gets \normalCompDist^{2^n}}\Big[ \exists i \in [2^n] : |v_i| > B \Big] \\
	&> \Pr_{v \gets \normalCompDist^{2^n}}\bigg[ \bigg( \norm{v} \geq \frac{\sqrt{2^n}}{2} \bigg) \land \big( \exists i \in [2^n] : |v_i| > B \big) \bigg] \\
	&= \Pr_{v \gets \normalCompDist^{2^n}}\bigg[ \norm{v} \geq \frac{\sqrt{2^n}}{2} \bigg] \cdot
	\Pr_{v \gets D}\Big[ \exists i \in [2^n] : |v_i| > B \Big] \\
	&\underset{\left( \substack{ \text{see proof of } \\ \text{Corollary \ref{lemma:gaussian_vectors_balanced}}} \right)}{\geq} \Big( 1 - e^{-\frac{2^n}{4}} \Big) \cdot \Pr_{v \gets D}\Big[ \exists i \in [2^n] : |v_i| > B \Big] \\
	&\underset{(2^n \geq 2)}{\geq} \Big( 1 - e^{-\frac{1}{2}} \Big) \cdot \Pr_{v \gets D}\Big[ \exists i \in [2^n] : |v_i| > B \Big] \\
	&> \left( \frac{3}{10} \right) \cdot \Pr_{v \gets D}\Big[ \exists i \in [2^n] : |v_i| > B \Big] \enspace .
	\end{align*}
	
	The above implies in particular,
	$$
	e^{-\secp} \cdot \left( \frac{10}{3} \right) > \Pr_{v \gets D}\Big[ \exists i \in [2^n] : |v_i| > B \Big] \enspace ,
	$$
	which implies our wanted inequality.
\end{proof}

\paragraph{Geometric Properties of a Vector and its Rounding.}
Besides statistical Properties of the Gaussian distribution, the below is a useful Lemma that will come in handy when we will want to use Lemma \ref{lemma:angular}.
Specifically we will use the below lemma when we'll want to argue that a vector and its rounding are close on the unit sphere (even in high dimensions), as long as the vector had sufficiently-large norm in the first place.

\begin{lemma} [Long Vectors with Coordinate-Bounded Difference are Close when Normalized] \label{lemma:long_vectors}
	Let $n \in \bbN$, $0 \leq \varepsilon \leq \frac{1}{5}$ and $u, v \in \bbC^{2^n}$ s.t.
	\begin{itemize}
		\item 
		$\norm{u} \geq \frac{\sqrt{2^n}}{2}$.
		
		\item 
		$\forall i \in [2^n]: \abs{u_i - v_i} \leq \varepsilon$.
	\end{itemize}
	Then we have the following lower bound on the inner product of the normalized vectors,
	$$
	\abs{\braket{ \hat{u} }{ \hat{v} }} \geq 1 - 18\cdot \varepsilon^2 \enspace .
	$$
\end{lemma}

\begin{proof}
	We prove a lower bound for $\abs{\braket{\hatu}{\hatv}}$ by giving an upper bound on the distance between the normalizations $\norm{\ket{\hatu} - \ket{\hatv}}$.
	Intuitively, the geometric argument goes as follows:
	\begin{itemize}
		\item
		The distance $\norm{\ket{\hatu} - \ket{\hatv}}$ is exactly the distance between the two (un-normalized) vectors $u$ and $\frac{\norm{u}}{\norm{v}} \cdot v$, multiplied by the factor $\frac{1}{\norm{u}}$.
		We later use the fact that $u$ is long, and thus the factor $\frac{1}{\norm{u}}$ will make the distance between the un-normalized vectors small.
		
		\item
		The distance from $u$ to $\frac{\norm{u}}{\norm{v}} \cdot v$ is bounded by the distance from $u$ to $v$, plus the distance from $v$ to $\frac{\norm{u}}{\norm{v}} \cdot v$.
		
		\item
		The distance from $u$ to $v$ is bounded as a function of $\varepsilon$, and the distance between $v$ and $\frac{\norm{u}}{\norm{v}} \cdot v$ is bounded by the fact that the norms $\norm{u}$, $\norm{v}$ are close.
	\end{itemize}

	Formally, assume w.l.o.g. that $u$ is the longer vector between the two, denote the difference vector $r := u - v$ and it follows that $\norm{r} \leq \sqrt{2^n} \cdot \varepsilon$.
	Additionally, from triangle inequality it follows that $\norm{v} \geq \left( \frac{1}{2} - \varepsilon \right) \cdot \sqrt{2^n}$.
	We have the following inequalities.
	
	\begin{align*}
	\norm{\ket{\hatu} - \ket{\hatv}} &=
	\norm{\frac{1}{\norm{u}} \cdot u - \frac{1}{\norm{v}}\cdot v} \\
	&= \frac{1}{\norm{u}}\cdot \norm{ u - \frac{\norm{u}}{\norm{v}}\cdot v } \\
	&= \frac{1}{\norm{u}}\cdot \norm{ v + r - \frac{\norm{u}}{\norm{v}}\cdot v } \\
	&\underset{(1)}{\leq} \frac{1}{\norm{u}}\cdot \left( \norm{ v - \frac{\norm{u}}{\norm{v}}\cdot v } + \norm{r} \right) \\
	&\leq \frac{1}{\norm{u}}\cdot \left( \norm{ v - \frac{\norm{u}}{\norm{v}}\cdot v }  + \varepsilon \cdot \sqrt{2^n} \right) \enspace ,
	\end{align*}
	
	where $(1)$ is due to the triangle inequality.
	
	Next, observe that the ratio $\frac{\norm{u}}{\norm{v}}$ is close to 1: it's at least 1 because we assume that $\norm{u} \geq \norm{v}$, and also,
	
	\begin{align*}
	\frac{\norm{u}}{\norm{v}} &=
	\frac{\norm{v + r}}{\norm{v}} \\
	&\leq \frac{\norm{v} + \norm{r}}{\norm{v}} \\
	&\leq \frac{\norm{v} + \varepsilon\cdot\sqrt{2^n}}{\norm{v}} \\
	&= 1 + \frac{\varepsilon\cdot\sqrt{2^n}}{\norm{v}} \\
	&\underset{\left( \norm{v} \geq \left( \frac{1}{2} - \varepsilon \right) \cdot \sqrt{2^n} \; , \; \frac{1}{4} \geq \varepsilon \right)}{\leq}
	1 + 4\cdot \varepsilon \enspace ,
	\end{align*}
	
	The above implies the following bound,
	$$
	\norm{ v - \frac{\norm{u}}{\norm{v}}\cdot v } = 
	\abs{ 1 - \frac{\norm{u}}{\norm{v}} }\cdot \norm{ v } \leq
	4\cdot \varepsilon \cdot \norm{v} \enspace ,
	$$
	
	and thus,
	\begin{align*}
	\frac{1}{\norm{u}}\cdot \left( \norm{ v - \frac{\norm{u}}{\norm{v}}\cdot v }  + \varepsilon \cdot \sqrt{2^n} \right)
	&\leq \frac{1}{\norm{u}}\cdot \left( 4\cdot \varepsilon \cdot \norm{v}  + \varepsilon \cdot \sqrt{2^n} \right) \\
	&= \frac{\norm{v}}{\norm{u}}\cdot 4\cdot \varepsilon + \frac{\varepsilon \cdot \sqrt{2^n}}{\norm{u}} \\
	&\underset{\Big( \substack{\norm{v}\leq \norm{u}, \\ \norm{u} \geq \sqrt{2^n}/2} \Big)}{\leq} 4\cdot \varepsilon + 2\cdot \varepsilon \\
	&= 6\cdot \varepsilon \enspace ,
	\end{align*}
	
	that is, $\norm{\ket{\hatu} - \ket{\hatv}} \leq 6\cdot \varepsilon$.
	
	Finally, recall the law of cosines that says that for vectors $x, y$:
	$$
	\norm{x - y}^2 =
	\norm{x}^2 + \norm{y}^2 - 2\cdot\norm{x}\cdot\norm{y}\cdot \cos(\Theta_{x, y}) =
	\norm{x}^2 + \norm{y}^2 - 2\cdot \braket{x}{y} \enspace ,
	$$
	which implies for unit vectors $\ket{\hat{x}}, \ket{\hat{y}}$,
	$$
	\implies \abs{\braket{\hat{x}}{\hat{y}}} = \abs{ 1 - \frac{\norm{\ket{\hat{x}} - \ket{\hat{y}}}^2}{2} } \enspace .
	$$
	In our case,
	$$
	\abs{\braket{\hatu}{\hatv}} = \abs{ 1 - \frac{\norm{\ket{\hatu} - \ket{\hatv}}^2}{2} } \underset{\left( \substack{\norm{\ket{\hatu} - \ket{\hatv}} \leq 6\cdot \varepsilon, \\ \varepsilon^2 \leq 1/18} \right)}{\geq}
	1 - 18\cdot \varepsilon^2 \enspace .
	$$
\end{proof}

\paragraph{Spherical Symmetry by Bound on Probability.}
Finally, an important property of the Gaussian distribution is that when we take a small number $\varepsilon$, the probability that a vector is of length $\varepsilon$-close (as a multiplication factor) to $\frac{\sqrt{2^n}}{2}$ is bounded by $\varepsilon\cdot N$.
A part of the reasoning in the proof of Theorem \ref{thm:main} will include moving from a non-spherically-symmetric distribution to a spherically symmetric one.
More specifically, the reason that the mentioned distribution is not symmetric is due to the fact that it's a conditional distribution, with the condition being a lower bound on the sampled vector length, but this lower bound changes from vector to vector and isn't uniform.
However, it is known that the changes on the lower bound between vectors are within a small multiplicative factor (which we think of as $\varepsilon$).

The below bound shows that changing the condition on a sampled vector's length by a bit (and thus making the condition uniform for all sampled vectors) yields a tiny change in statistical distance, as needed.

\begin{lemma}[Probability Mass of Vectors inside Thin Layer is Small] \label{lemma:layer_mass_small}
	Let $n \in \bbN$, $\secp \in \bbN$ (s.t. $\secp \geq 4$) and denote $\varepsilon := 2^{- n - \secp}$, $N := 2^n$.
	Then,
	$$
	\Pr_{u \gets \normalCompDist^{N}}\left[ \left( \frac{1}{2} - \varepsilon \right)\sqrt{N} \leq \norm{u} < \frac{1}{2}\cdot\sqrt{N} \right] <
	2^{-\secp}\cdot 6^{-N} \enspace .
	$$
\end{lemma}

\begin{proof}
	\znote{See example for equation alignment.}
	Denote $L := \biggl\{ u \in \bbR^{2N} : \left( \frac{1}{2} - \varepsilon \right)\sqrt{N} \leq \norm{u} < \frac{1}{2}\cdot\sqrt{N} \biggr\}$ and by $\mu$ the $2N$-dimensional Gaussian distribution measure, and we have,
	\begin{align*}
	\Pr_{u \gets \normalCompDist^{N}}\left[ \left( \frac{1}{2} - \varepsilon \right)\sqrt{N} \leq \norm{u} < \frac{1}{2}\cdot\sqrt{N} \right] &= 
	\Pr_{u \gets \normalDist^{2N}}\left[ \left( \frac{1}{2} - \varepsilon \right)\sqrt{N} \leq \norm{u} < \frac{1}{2}\cdot\sqrt{N} \right]\\
	& = 
	\mu\left( L \right) \enspace .
	\end{align*}
		
	We can bound the measure of $L$ by using the always-increasing exponential function and the multi-dimensional volume of $L$:
	\begin{align*}
	e^{\left( \frac{1}{4} - \varepsilon \right)N}\cdot \mu\left( L \right) &<
	e^{\left( \frac{1}{2} - \varepsilon \right)^2 N}\cdot \mu\left( L \right) \\
	&=
	\int_{L}e^{\left( \frac{1}{2} - \varepsilon \right)^2 N} \text{d}\mu \\
	&\leq
	\int_{L}e^{\norm{u}^2} \text{d}\mu \\
	& =
	\int_{L}e^{\norm{u}^2}\cdot e^{-\norm{u}^2/2}\cdot \left( 2\pi \right)^{-N} \text{d}u \\
	& = 
	\int_{L}e^{\frac{\norm{u}^2 - 2\cdot\ln(2\cdot \pi)N}{2}} \text{d}u \\
	& < 
	\int_{L}e^{\frac{N - 8\cdot\ln(2\cdot \pi)N}{8}} \text{d}u \\ 
	&=  
	\int_{L}e^{\frac{1 - 8\cdot\ln(2\cdot \pi)}{8} \cdot N} \text{d}u \\
	& < 
	e^{-\frac{17}{10} \cdot N}\cdot \int_{L}1\text{d}u \\
	& = 
	e^{-\frac{17}{10} \cdot N}\cdot \text{Vol}(L) \enspace ,
	\end{align*}
	which implies,
		\begin{align*}
	\mu\left( L \right) & <
	e^{-\left( \frac{1}{4} - \varepsilon \right)N}\cdot e^{-\frac{17}{10} \cdot N}\cdot \text{Vol}(L)\\
	& = 
	e^{- \left( \frac{39}{20} - \varepsilon \right) N} \cdot \text{Vol}(L)\\
	& \underset{(\varepsilon < 1/20)}{<} 
	e^{- \frac{38}{20} N} \cdot \text{Vol}(L) \enspace ,
	\end{align*}
	and it is left to bound the volume of $L$. \znote{The script under the inequality signs does not look good.} \omri{I thought this is the easiest way to explain the inequality. What is the usual design for this?}
	
	The volume of $L$ is a volume that can be calculated rather easily - this is simply the volume of the real sphere $S_{2N}\left(\frac{1}{2}\cdot\sqrt{N}\right)$ of $2N$ dimensions and radius $\frac{1}{2}\cdot\sqrt{N}$, minus the volume of the real sphere $S_{2N}\left( \left( \frac{1}{2} - \varepsilon \right)\sqrt{N} \right)$ of $2N$ dimensions and radius $\left( \frac{1}{2} - \varepsilon \right)\sqrt{N}$.
	Since there are known and convenient formulas for the volume of the $N$-ball, the calculation follows.
	\begin{align*}
	\text{Vol}(L) & = 
	\text{Vol}\left( S_{2N}\left( \frac{1}{2}\cdot \sqrt{N} \right) \right) - 
	\text{Vol}\left( S_{2N}\left( \left( \frac{1}{2} - \varepsilon \right)\sqrt{N} \right) \right) \\
	& = 
	\frac{\pi^{N}}{\Gamma(N + 1)}\cdot \left( \frac{1}{2}\cdot \sqrt{N} \right)^{2N} - \frac{\pi^{N}}{\Gamma(N + 1)}\cdot \left( \left( \frac{1}{2} - \varepsilon \right)\sqrt{N} \right)^{2N} \\
	& = 
	\frac{\pi^{N}\cdot N^N}{N!}\cdot \left( \left( \frac{1}{2} \right)^{2N} - \left( \frac{1}{2} - \varepsilon \right)^{2N} \right) \enspace .
	\end{align*}
	
	To bound the above, we'll give a lower bound for $\left( \frac{1}{2} - \varepsilon \right)^{2N}$:
	\begin{align}
	\left( \frac{1}{2} - \varepsilon \right)^{2N} &= 
	\sum_{i = 0}^{2N} {2N \choose i} (-1)^{i}\cdot \varepsilon^{i}\cdot \frac{1}{2^{2N - i}} \nonumber \\
	&=
	\frac{1}{2^{2N}} - 2N\varepsilon\cdot \frac{1}{2^{2N - 1}} + \sum_{i = 2}^{2N} {2N \choose i} (-1)^{i}\cdot \varepsilon^{i}\cdot \frac{1}{2^{2N - i}} \nonumber \\
	& \geq 
	\frac{1}{2^{2N}} - 2N\varepsilon\cdot \frac{1}{2^{2N - 1}} \label{eq:1} \enspace,
	\end{align}
	
	and we would like to show the inequality in Eq.~\eqref{eq:1} \znote{What is this "*"? Why not use equation reference?} \omri{How about now?}, that says that the last sum (that sums from $i = 2$ and up) is positive:
	\begin{align*}
	\sum_{i = 2}^{2N} {2N \choose i} (-1)^{i}\cdot \varepsilon^{i}\cdot \frac{1}{2^{2N - i}} &= 
	\sum_{i = 2, 4, 6, \cdots, 2N - 2} \left( {2N \choose i}\cdot \varepsilon^{i}\cdot \frac{1}{2^{2N - i}} -
	2\cdot{2N \choose i + 1}\cdot\varepsilon^{i + 1}\cdot \frac{1}{2^{2N - i}} \right) + \varepsilon^{2N} \\
	&= 
	\sum_{i = 2, 4, 6, \cdots, 2N - 2} \left( \varepsilon^{i}\cdot \frac{1}{2^{2N - i}}\cdot \left( {2N \choose i} -  2\cdot{2N \choose i + 1}\cdot\varepsilon \right) \right) + \varepsilon^{2N} \enspace ,
	\end{align*}
	and for all $i = 2, 4, 6, \cdots, 2N - 2$ we have,
	\begin{align*}
		{2N \choose i} -  2\cdot{2N \choose i + 1}\cdot\varepsilon & = 
	\prod_{j = 0}^{i - 1}\frac{2N - j}{i - j} - 2\varepsilon\cdot \prod_{j = 0}^{i}\frac{2N - j}{i + 1 - j} \\
	&= 
	\prod_{j = 0}^{i - 1}\frac{2N - j}{i - j} - 2\varepsilon\cdot \left( 2N - i \right) \prod_{j = 0}^{i - 1}\frac{2N - j}{i + 1 - j} \\ 
	&\geq
	\prod_{j = 0}^{i - 1}\frac{2N - j}{i - j} - 4N\varepsilon\cdot \prod_{j = 0}^{i - 1}\frac{2N - j}{i + 1 - j} \\
	& \geq
	\left( \prod_{j = 0}^{i - 1}\frac{2N - j}{i - j} \right) \cdot \left( 1 - 4N\varepsilon \right)\\
	& \underset{(1 \geq 2^{-\secp + 2} = 4N\varepsilon)}{\geq} 0 \enspace.
	\end{align*}
	The last inequality implies the inequality \ref{eq:1}, which in turn implies the lower bound,
	$$
	\left( \frac{1}{2} - \varepsilon \right)^{2N} \geq 
	\frac{1}{2^{2N}} - 4N\varepsilon\cdot \frac{1}{2^{2N}} \enspace.
	$$
	
	Finally, we obtain the following:
	\begin{align*}
	\mu\left( L \right) & <
	e^{- \frac{38}{20} N} \cdot \text{Vol}(L) \\
	& = 
	e^{- \frac{38}{20} N} \cdot \frac{\pi^{N}\cdot N^N}{N!}\cdot \left( \left( \frac{1}{2} \right)^{2N} - \left( \frac{1}{2} - \varepsilon \right)^{2N} \right)\\
	& \leq
	e^{- \frac{38}{20} N} \cdot \frac{\pi^{N}\cdot N^N}{N!}\cdot 4N\varepsilon\cdot \frac{1}{2^{2N}} \\
	& =
	2^{-\secp + 2}\cdot \frac{\pi^{N}\cdot N^N}{e^{\frac{38}{20} N} \cdot 4^{N} \cdot N!} \\
	& \underset{(\text{Stirling})}{\leq}
	2^{-\secp + 2}\cdot \frac{\pi^{N}\cdot N^N}{e^{\frac{38}{20} N} \cdot 4^{N} \cdot \sqrt{2\pi N}\left( \frac{N}{e} \right)^N} \\
	& = 
	2^{-\secp + 2}\cdot \frac{\pi^{N}}{e^{\frac{18}{20} N} \cdot 4^{N} \cdot \sqrt{2\pi N}} \\
	& \underset{\left( 6 < e^{\frac{18}{10}} \right)}{<}
	2^{-\secp}\cdot 6^{-N} \enspace .\qedhere
	\end{align*}
\end{proof}

The Lemma about the tiny probability that a Gaussian vector lands inside the layer that's $\varepsilon$-near $\frac{\sqrt{2^n}}{2}$ implies that it is also the case when it is sampled from the conditional distribution of balanced vectors.

\begin{corollary}[Balanced Gaussian Vectors are Almost Always Outside of the Thin Layer] \label{cor:balanced_vectors_not_in_layer}
	Let $n \in \bbN$, $4 \leq \secp \in \bbN$, $\varepsilon := 2^{-n-\secp}$, $B \geq 2\sqrt{n + \secp}$ and let $D$ be the conditional distribution of $u \gets \normalCompDist^{N}$ with the conditions that $\norm{R_{(\varepsilon, B)}(u)} \geq \frac{\sqrt{N}}{2}$ and $\forall i \in [2^n] : |u_i| \leq B$.
	Then,
	$$
	\Pr_{u \gets D}\left[ \norm{u} < \frac{\sqrt{2^n}}{2} \right] < 4\cdot2^{-\secp} \cdot 6^{-2^n} \enspace .
	$$
\end{corollary}

\begin{proof}
	We have,
	\begin{align*}
	2^{-\secp} \cdot 6^{-2^n}
	&\underset{(\text{Lemma \ref{lemma:layer_mass_small}})}{>}
	\Pr_{u \gets \normalCompDist^{2^n}}\left[ \left( \frac{1}{2} - \varepsilon \right)\sqrt{2^n} \leq \norm{u} \frac{1}{2}\cdot\sqrt{2^n} \right] \\
	&\geq \Pr_{u \gets \normalCompDist^{2^n}}\left[ \left( \norm{R_{(\varepsilon, B)}(u)} \frac{\sqrt{2^n}}{2} \right) \land \left( \forall i \in [2^n] : |u_i| \leq B \right) \land \left( \norm{u} < \frac{\sqrt{2^n}}{2} \right) \right] \\
	&= \Pr_{u \gets \normalCompDist^{2^n}}\left[ \left( \norm{R_{(\varepsilon, B)}(u)} \frac{\sqrt{2^n}}{2} \right) \land \left( \forall i \in [2^n] : |u_i| \leq B \right) \right] \cdot
	\Pr_{u \gets D}\left[  \norm{u} < \frac{\sqrt{2^n}}{2} \right] \\
	&\geq \Pr_{u \gets \normalCompDist^{2^n}}\left[ \left( \norm{u} \geq \frac{\sqrt{2^n}}{2} \right) \land \left( \forall i \in [2^n] : |u_i| \leq B \right) \right] \cdot
	\Pr_{u \gets D}\left[  \norm{u} < \frac{\sqrt{2^n}}{2} \right] \\
	&\underset{\text{(Corollary \ref{lemma:gaussian_vectors_balanced})}}{>}
	\Big( 1 - \left( e^{-\frac{2^n}{4}} + e^{-\secp} \right) \Big) \cdot \Pr_{u \gets D}\left[  \norm{u} < \frac{\sqrt{2^n}}{2} \right] \\
	&\underset{(2^n \geq 2 \; , \; \secp \geq 3)}{\geq}
	\Big( 1 - \left( e^{-\frac{1}{2}} + e^{-3} \right) \Big) \cdot \Pr_{u \gets D}\left[  \norm{u} < \frac{\sqrt{2^n}}{2} \right] >
	\left( \frac{3}{10} \right) \cdot \Pr_{u \gets D}\left[  \norm{u} < \frac{\sqrt{2^n}}{2} \right] \enspace .
	\end{align*}
	The above implies in particular,
	$$
	2^{-\secp} \cdot 6^{-2^n} \cdot \left( \frac{10}{3} \right) > \Pr_{u \gets D}\left[  \norm{u} < \frac{\sqrt{2^n}}{2} \right] \enspace ,
	$$
	which implies our wanted inequality.
\end{proof}

	

	\bibliographystyle{alpha} 
	\bibliography{PseudoRandomQuantum}

\newpage
\pagestyle{plain}
\pagenumbering{roman}

	

\end{document}